\documentclass[a4paper,twocolumn,11pt,accepted=2024-01-22]{quantumarticle}
\pdfoutput=1 
\usepackage{amsmath}
\usepackage{amsthm}
\usepackage{amssymb}
\usepackage{bm}
\usepackage{dsfont}
\usepackage{mathtools}
\usepackage{hyperref}
\usepackage{braket}
\usepackage{color}
\usepackage{multirow}
\usepackage{mathtools}

\def\>{\rangle}
\def\<{\langle}

\newcommand\numberthis{\addtocounter{equation}{1}\tag{\theequation}}

\DeclareMathOperator*{\Tr}{Tr}
\DeclareMathOperator*{\Id}{\mathds{1}}

\begin{document}

\title{Hamiltonian variational ansatz without barren plateaus}

\author{Chae-Yeun Park}
\author{Nathan Killoran}
\affiliation{Xanadu, Toronto, ON, M5G 2C8, Canada}


\newtheorem{proposition}{Proposition}
\newtheorem{example}{Example}
\newtheorem{lemma}{Lemma}
\newtheorem{theorem}{Theorem}
\newtheorem{observation}{Observation}
\newtheorem{corollary}{Corollary}
\newtheorem{conjecture}{Conjecture}

\begin{abstract}
Variational quantum algorithms, which combine highly expressive parameterized quantum circuits (PQCs) and optimization techniques in machine learning, are one of the most promising applications of a near-term quantum computer.
Despite their huge potential, the utility of variational quantum algorithms beyond tens of qubits is still questioned.
One of the central problems is the trainability of PQCs. The cost function landscape of a randomly initialized PQC is often too flat, asking for an exponential amount of quantum resources to find a solution.
This problem, dubbed \textit{barren plateaus}, has gained lots of attention recently, but a general solution is still not available.
In this paper, we solve this problem for the Hamiltonian variational ansatz (HVA), which is widely studied for solving quantum many-body problems.
After showing that a circuit described by a time-evolution operator generated by a local Hamiltonian does not have exponentially small gradients,
we derive parameter conditions for which the HVA is well approximated by such an operator.
Based on this result, we propose an initialization scheme for the variational quantum algorithms and a parameter-constrained ansatz free from barren plateaus.
\end{abstract}

\maketitle

\section{Introduction}

Recent experimental progress in controlling quantum systems has demonstrated quantum advantages in sampling tasks~\cite{arute2019quantum,zhong2020quantum,madsen2022quantum},
and near-term quantum computers with hundreds of noisy qubits are emerging~\cite{preskill2018quantum}.
Variational quantum algorithms (VQAs) are one of the most promising applications of these near-term quantum computers.
By combining highly expressive parameterized quantum circuits (PQCs) and well-established parameter optimization techniques from machine learning (ML), 
VQAs are relevant for many important problems, including combinatorial optimizations~\cite{farhi2014quantum}, finding the ground state of a many-body Hamiltonian~\cite{peruzzo2014variational,wecker2015progress,kandala2017hardware,hadfield2019quantum}, and learning probability distributions~\cite{schuld2015introduction,biamonte2017quantum,schuld2019quantum,liu2021rigorous} (see Ref.~\cite{cerezo2021variational} for a recent review).

VQAs solve a problem by optimizing a cost function typically defined by the expectation value of a target-problem specific observable.
However, this optimization task can be challenging since the cost function landscapes are often too flat~\cite{mcclean2018barren,cerezo2021cost}.
This phenomenon, dubbed \textit{barren plateaus}, is characterized by the fact that all gradient components are exponentially small with the number of qubits when parameters are randomly sampled.
Given that barren plateaus are expected to be prevalent for sufficiently expressive ans\"{a}tze~\cite{holmes2022connecting},
\textit{trainability} of PQCs beyond tens of qubits is still an open question.

The issue of vanishing gradients is not entirely new, though.
Classical neural networks also suffered a similar vanishing gradient problem, but
theoretical and numerical advances have shown that clever neural network architectures~\cite{hochreiter1997long,glorot2011deep} or better initialization methods~\cite{glorot2010understanding,he2015delving} can sufficiently suppress the problem.
Likewise, recent studies explored quantum circuit ans\"{a}tze without barren plateaus~\cite{zhang2020toward,volkoff2021large,pesah2021absence,liu2022mitigating}, as well as initialization techniques that provide large gradients~\cite{grant2019initialization,jain2021graph,zhang2022gaussian,mele2022avoiding,rudolph2023synergistic}.
Still, it is unclear how useful barren-plateau-free ans\"{a}tze are for solving complex problems.
Also, proposed initialization methods mostly rely on heuristics and do not provide strong arguments for why such parameters should yield a large gradient.

In this paper, we resolve these issues for the Hamiltonian variational ansatz (HVA) by proposing a novel parameter initialization technique.
The HVA~\cite{wecker2015progress,hadfield2019quantum} is widely studied for solving the ground state of a many-body Hamiltonian since it can encode adiabatic evolution.
However, the HVA is still subject to the barren plateau problem~\cite{wiersema2020exploring,larocca2022diagnosing}.
Even though several initialization methods based on pre-training have been proposed to overcome this problem~\cite{mele2022avoiding,rudolph2023synergistic},
those methods not only require additional classical or quantum resources but rely on heuristics developed based on numerical results for less than $20$ qubits.
In contrast, our initialization scheme simply adds a constraint to the parameters and is free from additional computational resources.
Moreover, we provide a rigorous argument for why this scheme yields large gradients, supported by extensive numerical results up to 28 qubits.
We further propose an ansatz that imposes the constraint throughout the optimization process.
Such an ansatz is expressive enough for variational time evolution~\cite{li2017efficient,yuan2019theory,cirstoiu2020variational,lin2021real,keever2022classically}, with the benefit that the ansatz is free from barren plateaus.

The remainder of the paper is organized as follows.
After briefly introducing the problem and related concepts in Sec.~\ref{sec:preliminaries}, 
we show that the gradient does not decay exponentially when a circuit is described by local Hamiltonian evolution in Sec.~\ref{sec:var_grad_ham_dynamics}.
In Sec.~\ref{sec:circuit_local_ham_approx}, we find a parameter condition for which the HVA approximates to local Hamiltonian dynamics. We thus prove that a parameter regime for constant gradient magnitudes exists.
We then introduce an initialization method based on our proof and numerically compare it to other known parameter initialization techniques in Sec.~\ref{sec:numerical_comparisions_init}.
We summarize our results with concluding remarks in Sec.~\ref{sec:conclusion}.

\begin{figure}[t]
    \centering
    \includegraphics[width=0.9\linewidth]{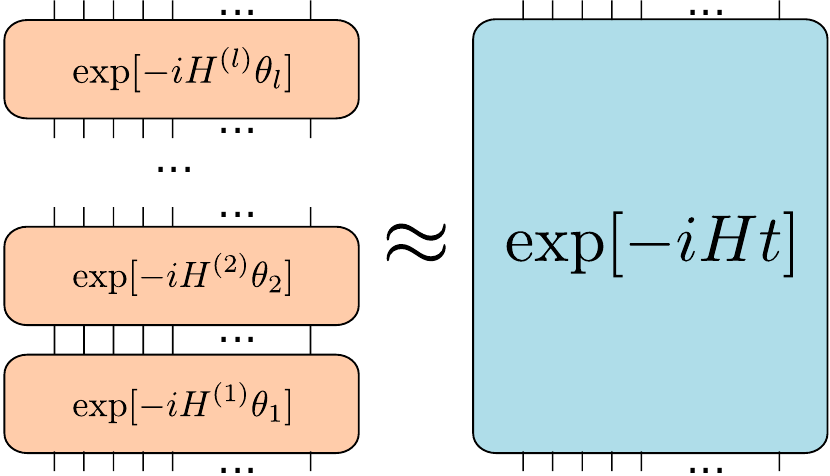}
    \caption{We find a parameter constraint such that layers of Hamiltonian evolution in the HVA (left) approximate to the time evolution under a single local Hamiltonian (right). Using the dynamical properties of local Hamiltonians, we argue that the HVA has large gradients.}
    \label{fig:hva_schematic_diagram}
\end{figure}

\section{Preliminaries}
\label{sec:preliminaries}
We consider a PQC for $N$ qubits with $l$ total layers, given by
\begin{align}
    U(\pmb{\theta}) = U_l(\theta_l)\cdots U_1(\theta_1),
\end{align}
where $\pmb{\theta}=(\theta_1,\cdots,\theta_l)$ is a vector of all parameters and $U_n(\theta)=e^{-i \theta G_n}$ is a unitary gate generated by $G_n$.
In VQAs, a cost function is typically given by
\begin{align}
    C(\pmb{\theta}) = \Tr[O U(\pmb{\theta})\rho_0 U^\dagger(\pmb{\theta})],
\end{align}
where $O$ is a Hermitian operator.
The cost function is then optimized with gradient-based methods. Direct computation of the gradient yields
\begin{align}
    \partial_n C = \frac{\partial C}{\partial \theta_n} = i \Tr[U_R \rho_0 U_R^{\dagger} [G_n, U_L^{\dagger}O U_L]] \label{eq:PQC}, 
\end{align}
where $U_R = U_{n-1} \cdots U_1$, $U_L=U_l \cdots U_n$, and $\rho_0$ is the initial state of the circuit.

For classes of PQCs, which form a 1-design, the gradient is unbiased for a given parameter set (i.e., $\mathbb{E}_{\pmb{\theta}}[\partial_n C] = 0$).
In this case, one can use the variance to quantify the magnitudes of gradients, which is given by
\begin{align}
    \mathrm{Var}[\partial_n C] &= \int d \mu(\pmb{\theta}) \left( \frac{\partial C}{\partial \theta_n} \right)^2 \nonumber \\
    &= - \int d \mu(\pmb{\theta}) \Tr[U_R \rho_0 U_R^{\dagger} [G_n, U_L^{\dagger}O U_L]]^2. \label{eq:var_g}
\end{align}

In the typical barren plateau scenario~\cite{mcclean2018barren,cerezo2021cost}, this quantity becomes close to $\mathcal{O}(1/D^2)$~\footnote{See Appendix~\ref{app:big-O} for the definition of big-$O$ and related notations}, which is the value evaluated under the assumption that $U_{R}$ or $U_L$ is a unitary 2-design.
Here, $D$ is the total dimension of the Hilbert space, which is $2^N$ for a system with $N$ qubits.
Hence, the variance decays exponentially with the number of qubits, which implies that the gradient is exponentially small for most values of the parameters (can be rigorously proven by Chebyshev's inequality).
Even though it is possible to optimize the cost function using a small gradient in principle, running the algorithm in real quantum hardware is extremely inefficient as estimating the gradient requires an exponential number of shots.

Next, we introduce the HVA. The HVA~\cite{wecker2015progress,hadfield2019quantum} is a natural ansatz for solving quantum many-body Hamiltonians. After decomposing a given Hamiltonian into $q$ terms $H=\sum_{j=1}^q c_j H^{(j)}$, where $\{c_j\}$ are real coefficients, the HVA is constructed as
\begin{align}
    &\ket{\psi(\{\theta_{i,j}\})}  \nonumber \\
    &= \prod_{i=p}^1 \bigl[e^{-i H^{(q)}\theta_{i,q}} \cdots e^{-i H^{(2)} \theta_{i,2}} e^{-i H^{(1)} \theta_{i,1}}\bigr] \ket{\psi_0}, \label{eq:HVA}
\end{align}
where $\ket{\psi_0}$ is a quantum state that can be easily prepared.
The ansatz consists of $p$ blocks, each containing $q$ layers.
Thus, the ansatz has a total of $l=pq$ layers. 
We also use the notation $\theta_{a}$ and $U_a$ to denote $\theta_{i,j}$ and $U_a = e^{-iH^{(j)}\theta_{i,j}}$ where $a = (i,j)$,
which enables us to interpret the HVA as a PQC given by Eq.~\eqref{eq:PQC}.

Throughout the paper, we restrict $H^{(j)}$ to be a $k$-local Hamiltonian in a given lattice for a constant $k$, i.e., each term in the Hamiltonian acts on at most $k$ \textit{geometrically nearby sites} in a given lattice.
This condition is satisfied for most of the many-particle spin-$1/2$ Hamiltonians.
For example, we decompose the one-dimensional transverse-field Ising model $\mathcal{H} = -\sum_iZ_i Z_{i+1} + h X_i$ into $H^{(1)} = -\sum_i Z_i Z_{i+1}$ and $H^{(2)}= - \sum_i X_i$. Then $H^{(1)}$ is 2-local and $H^{(2)}$ is 1-local.

This ansatz is powerful for solving the ground state of $H$, as it can encode the adiabatic evolution of the Hamiltonian~\cite{hadfield2019quantum}.
Despite the usefulness of the ansatz, however, training the HVA turned out to be non-trivial.
Some numerical studies have observed that the gradients decay exponentially with the system size~\cite{wiersema2020exploring,larocca2022diagnosing}, although
the magnitudes of gradients of the HVA are larger than one expects from a unitary 2-design~\cite{wiersema2020exploring}.

In this paper, we consider the case where the HVA is well approximated by time evolution under a local Hamiltonian, i.e., there are local Hamiltonians $H_L, H_R$ such that $U_L \approx e^{-i H_L t_L}$ and  $U_R \approx e^{-i H_R t_R}$ for some $t_L,t_R \geq 0$ (see Fig.~\ref{fig:hva_schematic_diagram}).
With this assumption, we provide strong analytic and numerical arguments that the gradient magnitudes, Eq.~\eqref{eq:var_g}, only decay at most polynomially in the number of qubits.
Although this assumption seems unrealistic, we later show that the HVA with a certain parameter restriction can satisfy this condition.

\section{Magnitudes of gradients in Hamiltonian dynamics}
\label{sec:var_grad_ham_dynamics}

In this section, we study the scaling of the gradient when the circuit is given by the time evolution under a time-independent local Hamiltonian.
We show that gradients in such circuits do not decay exponentially in both extreme regimes: short- and long-time evolution.

For short-time evolution, we prove a rigorous bound on the time that the gradient preserves its initial magnitudes.
Thus, a circuit have large gradients for a proper initial state.
On the other hand, for long-time evolution, we combine the universality of quantum thermalization~\cite{deutsch1991quantum,srednicki1994chaos,rigol2008thermalization} and our numerical results to argue that the gradient does not decay exponentially.

\subsection{Gradient scaling for short-time evolution}
In this subsection, assuming that (1) a circuit with $N$ qubits is given by $e^{-iHt}$ for a local Hamiltonian $H$ and (2) the initial state has a large gradient with a value of $\Theta(1)$, we prove that there exists $t_c = \Theta(1/N)$ such that the circuit maintains the large gradient when the total evolution time is less than $t_c$.
Our main result is the following proposition:

\begin{proposition}[Quantum speed limit of gradients] \label{prop:short_time_grad_var}
	For the HVA, the gradient of the cost function is given by 
	\begin{align}
	\partial_{n,m}C = \frac{\partial C}{\partial \theta_{n,m}} = i\Tr[U_R \rho_0 U_R^\dagger [H^{(m)}, U_L^\dagger O U_L]]
	\end{align}
	where $U_R = e^{-iH^{(m-1)}\theta_{n,m-1}} \cdots e^{-iH^{(1)} \theta_{1,1}}$ and $U_L = e^{-iH^{(q)}\theta_{p,q}}\cdots e^{-iH^{(m)}\theta_{n,m}}$.
	Assume that the gradient component, $\partial_{n,m}C$, is non-zero when the circuit is identity, i.e., $|\Tr[\rho_0 [H^{(m)},O] ]| > 0$, and there are local Hamiltonians $H_L$, $H_R$ such that $U_L = e^{-i H_L t_L}$ and $U_R = e^{-i H_R t_R}$ for some $t_R,t_L \geq 0$. Then, 
	\begin{align}
	    \Bigl| \frac{\partial C}{\partial \theta_{n,m}} \Bigr| \geq |\Tr[\rho_0 [H^{(m)},O] ]|/2 
	\end{align}
	for $t_R + t_L \leq t_c :=|\Tr[\rho_0 [H^{(m)}, O]]|/(4 K C)$, where $K = \max\{\Vert H_R \Vert, \Vert H^{(m)} \Vert\}$, $C = \max\{\Vert [H^{(m)}, O] \Vert, \Vert [H_L, O] \Vert \}$, and $\Vert \cdot \Vert$ is the operator norm.
\end{proposition}

\begin{proof}
Let 
\begin{align}
    &A(t_1, t_2) \nonumber \\
    &= i\Tr [ e^{-iH_R t_1} \rho_0 e^{i H_R t_1} [H^{(m)}, e^{iH_L t_2} O e^{-iH_L t_2}]].
\end{align}
Then 
\begin{align}
    &|A(t_R, t_L) - A(0,0)| \nonumber \\
    &\qquad \leq \int_{0}^{t_R} dt_1 \Bigl| \frac{\partial A(t_1, 0)}{\partial t_1}  \Bigr| + \int_{0}^{t_L} dt_2 \Bigl| \frac{\partial A(t_R, t_2)}{\partial t_2}  \Bigr|.
\end{align}
We further have
\begin{align*}
    \Bigl| \frac{d A(t_1,0)}{\partial t_1} \Bigr| &= \Bigl|\Tr\bigl\{ [H_R, \rho_0(t_1)] [H^{(m)}, O] \bigr\} \Bigr| \\
    &\leq 2 \Vert H_R \Vert \Vert [H^{(m)}, O] \Vert \leq 2KC, \numberthis
\end{align*}
and 
\begin{align*}
    &\Bigl| \frac{d A(t_R,t_2)}{\partial t_2} \Bigr| \\
    &= \Bigl|\Tr\bigl\{ \rho_0(t_R) [H^{(m)}, [H_L, e^{i H_L t}Oe^{-i H_L t}]] \bigr\} \Bigr|\\
    &\leq 2 \Vert H^{(m)} \Vert \Vert [H_L, O] \Vert \leq 2KC, \numberthis
\end{align*}
where $\rho_0(t) = e^{-iH_R t} \rho_0 e^{i H_R t}$.

Integrating both sides, we have
\begin{align}
    |A(t_R, t_L) - A(0, 0)| \leq 2 K C(t_R + t_L). \label{eq:partial_deriv_bound}
\end{align}
By entering $t_R+t_L \leq t_c = |A(0, 0)|/(4KC)$, we obtain $|A(t_R, t_L) - A(0, 0)| \leq |A(0, 0)|/2$, i.e.,
\begin{align}
    &A(0,0) - |A(0,0)|/2 \leq A(t_R, t_L) \nonumber \\
    &\qquad \leq A(0,0) + |A(0,0)|/2.
\end{align}
We obtain the desired inequality as $A(t_R,t_L) \geq A(0,0)/2 > 0$, if $A(0,0) > 0$, and $A(t_R,t_L) \leq A(0,0)/2 < 0$, otherwise.
\end{proof}

Let us assume that all of $H^{(m)}$, $H_L$, and $H_R$ are $k$-local Hamiltonians, where each term acts at most $k$ nearby sites in a given lattice for a constant $k$, and $O$ is a local operator acting on at most a constant number of sites.
Under this assumption, which is the case we consider in this paper, we have $t_c = \Theta(1/N)$ when $|\Tr[\rho_0 [H^{(m)}, O]]| = \Theta(1)$.
We prove this fact in the rest of the subsection.

For any $k$-local Hamiltonian $H$, we can write
\begin{align}
H = \sum_{i \in \Lambda} h_i
\end{align}
where $\Lambda=\{1,\cdots,N\}$ is the collection of all sites, and $h_i$ is an operator supported by $k$ sites centered at $i$.
Formally, we write the support of $h_i$ (a set of sites $h_i$ acts on) as 
\begin{align}
\mathrm{supp}(h_i)=\{j \in \Lambda: \mathrm{dist}(i, j) \leq k\}
\end{align}
where $\mathrm{dist}(i, j)$ is the distance between two sites in the given lattice.
We thus have 
\begin{align}
    \Vert H \Vert \leq N \max_i \Vert h_i \Vert, \quad \Vert [H, O] \Vert \leq 2 s \Vert O \Vert \max_i \Vert h_i \Vert \label{eq:local_ham_norms}
\end{align}
for a local operator $O$. Here,
\begin{align}
s = |\{i \in \Lambda: \mathrm{dist}(i, O) \leq k \}|
\end{align}
is a constant for a given lattice, where $\mathrm{dist}(i, O) = \min_{j \in \mathrm{supp}(O)} \mathrm{dist}(i,j)$ and $\mathrm{supp}(O) \subset \Lambda$ is the support of $O$.
Given that $\Vert h_i \Vert $, $\Vert O \Vert$ are bounded by a constant for a spin system, and $s$ is a constant for a finite-dimensional lattice, we have
\begin{align}
\Vert H \Vert = \mathcal{O}(N), \quad \Vert [H, O]\Vert = \mathcal{O}(1)
\end{align}
for any $k$-local Hamiltonain $H$.
We also note that physical Hamiltonians must have $\Vert H \Vert = \Theta(N)$, which is a necessary condition to be thermodynamically well-defined.
Therefore, we obtain $t_c = \Theta(1/N)$ if the circuit has a large initial gradient component, i.e., if there exists $m$ such that $|\Tr[\rho_0 [H^{(m)},O] ]| = \Theta(1)$.

For example, when $H^{(m)}=\sum_i Y_i$, $O=Z_1$, we have $i\Tr[\rho_0 [H^{(m)},O]] = -2$ for $\ket{\psi_0} = |+\rangle^{\otimes N}$.
Moreover, we have $K = \Theta(N)$ and $C = \Theta(1)$ for Proposition~\ref{prop:short_time_grad_var}, when $H_R$ and $H_L$ are also $k$-local, which implies $t_c = \Theta(1/N)$.
Therefore, the gradient component is $\Theta(1)$ for all $t \leq t_c$.

While we mostly consider geometrically local Hamiltonians in this paper (i.e., local in a finite-dimensional lattice), our result can be extended to Hamiltonians defined on a general (hyper)graph.
Such a complex Hamiltonian appears when a fermionic Hamiltonian is translated to a spin Hamiltonian using, e.g., the Jordan–Wigner transformation (see, e.g., Ref.~\cite{wecker2015progress}).
These Hamiltonians can have smaller $t_c$ because $C=\max\{\Vert [H^{(m)}, O] \Vert, \Vert [H_L, O] \Vert \}$ in Proposition~\ref{prop:short_time_grad_var} may scale linearly with $N$.
For example, consider a Hamiltonian $H$, each term of which acts on all sites.
Namely, we have $H$ given as
\begin{align}
    H = \sum_{i=1}^N h_i,
\end{align}
where each $h_i$ is a Pauli string acting non-trivially on all sites, i.e., $h_i \in \{X, Y, Z\}^{\otimes N}$.
We still have $\Vert H \Vert = \Theta(N)$ for this Hamiltonian.
However, for a local operator, $O$, acting on site $i$, we can have $\Vert [O, H] \Vert = \Theta(N)$.
Thus, applying Proposition~\ref{prop:short_time_grad_var} to this Hamiltonian yields $t_c = \Theta(1/N^2)$ instead of $\Theta(1/N)$.

\subsection{Gradient scaling for long-time evolution}\label{sec:grad_scaling_long_time}

Next, we consider long-time evolution. Following usual arguments for equilibration~\cite{reimann2008foundation,rigol2008thermalization,linden2009quantum,short2011equilibration,gogolin2016equilibration}, we assume that a Hamiltonian $H$ follows the non-degenerate energy-gap condition: $E_i - E_j = E_k - E_l$ iff $i=k$ and $j=l$; or $i=j$ and $k=l$, where $E_i$ is the $i$-th eigenvalue of $H$ with the corresponding eigenvector $\ket{E_i}$.
With an additional assumption that the Hamiltonian thermalizes~\cite{deutsch1991quantum,srednicki1994chaos,rigol2008thermalization}, in the sense that the observable after equilibration gives a similar value to the thermal average,
it is known that the second moment of the Hamiltonian evolution behaves differently from a unitary 2-design~\cite{huang2019finite}.

Precisely, Huang et al.~\cite{huang2019finite} considered the saturated value of the out-of-time correlator (OTOC), a widely used measure for detecting quantum chaos. For local Hermitian operators, $O_i$ and $O_j$ acting on sites $i$ and $j$, respectively, the OTOC is defined by
\begin{align}
    \mathrm{OTOC}(O_i, O_j) &:= \Tr\Bigl[ \rho_0 \bigl( U O_i U^{\dagger} O_j\bigr)^2 \Bigr].
\end{align}
Here, $\rho_0$ is the initial state, and $U$ is a unitary operator determining the time evolution of the system.

One often considers the infinite temperature initial state given by $\rho_0 = \Id/2^N$ for a system with $N$ qubits.
When our unitary operator $U$ forms a 2-design, we obtain
\begin{align}
    \mathrm{OTOC}(O_i, O_j) &= \frac{1}{2^N}\Tr\Bigl[U O_i U^{\dagger} O_j U O_i U^{\dagger} O_j \Bigr] \\
                            & \xrightarrow{\text{Haar } U} -\frac{2^N}{2^{2N}-1},
\end{align}
for traceless $O_i$ and $O_j$ (e.g., local Pauli operators; see, e.g., Ref.~\cite{roberts2017chaos} for a proof).
Namely, $\mathrm{OTOC}(O_i, O_j)$ scales inverse \textit{exponentially} with $N$ in this case.
On the other hand, $\mathrm{OTOC}(O_i, O_j)$ scales only inverse \textit{polynomially} with the system size for local Hamiltonian evolution, i.e., when $U=e^{-iHt}$ for a local Hamiltonian $H$~\cite{huang2019finite}. Such a huge difference mainly comes from the fact that $U=e^{-iHt}$ conserves the energy, i.e., $[H,U]=0$.

\begin{figure}
    \centering
    \includegraphics[width=0.95\linewidth]{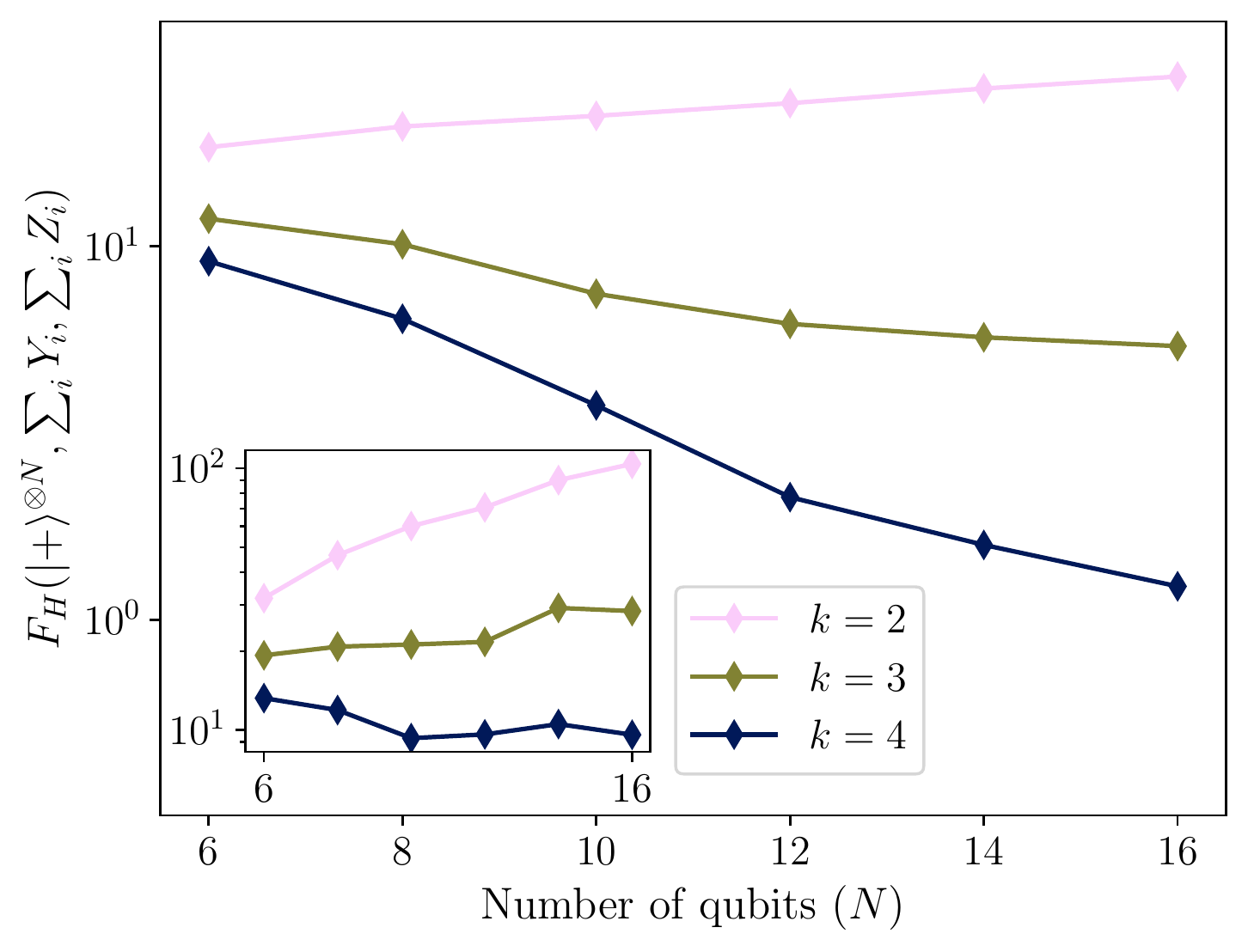}
    \caption{The lower bound $F_H(\ket{\psi}, H^{(1)}, O)$ for $\ket{\psi} = |+\rangle^{\otimes N}$, $H^{(1)}=\sum_i Y_i$, and $O = \sum_i Z_i$ averaged over $2^{10}$ randomly generated $k$-local Hamiltonians. Inset: The same results but for Hamiltonians with time-reversal symmetry.}
    \label{fig:random_ham_f}
\end{figure}

Similarly, one might expect that the variance of gradients saturates to a value that scales only inverse polynomially with $N$ for local Hamiltonian evolution. 
To see that this is the case, we compute the square of the first element of the gradient (for $\theta_{1,1}$) when $U_L = e^{-iH_L t}$ and the initial state is given as a pure state $\rho_0 = \ket{\psi_0} \bra{\psi_0}$.
Then we obtain a lower bound of $(\partial_{\theta_{1,1}}C)^2 = -\braket{\psi_0 | [H^{(1)}, O(t)]| \psi_0}^2$ where $O(t) = e^{iH_L t} O e^{-iH_L t}$ as follows:

\begin{proposition}\label{prop:lte_grad}
Assume that $H_L$ satisfies the non-degenerate energy-gap condition.
Then the long-time average of $-\langle \psi_0 | [H^{(1)}, O(t)] | \psi_0\rangle ^2$ is lower bounded by $F_{H_L}(\ket{\psi_0}, H^{(1)}, O)$.
Here, $F_{H}(\ket{\psi}, G, O)$ is a function given by
\begin{align*}
&F_{H}(\ket{\psi}, G, O) \\
&\qquad= 2 \sum_{ijkl} C_i^* G_{ij} O_{jk} |C_k|^2 O_{kj} G_{jl} C_l \\
&\qquad - \sum_{ijkl} C_i^* O_{ij} G_{jk} C_k C_j^* O_{ji} G_{il} C_l \\
&\qquad - \sum_{ijkl} C_i^* G_{ij} O_{jk} C_k C_l^* G_{lk}O_{kj} C_j, \numberthis
\end{align*}
where $C_i = \braket{E_i|\psi_0}$, $G_{ij} = \braket{E_i | G | E_j}$, and $O_{jk} = \braket{E_j | O | E_k}$.
Here, $\ket{E_i}$ is the $i$-th eigenstate of $H$.
\end{proposition}

A proof can be found in Appendix~\ref{app:long-time-avg-grad-ham-evol}.
Unfortunately, we cannot use techniques to analytically compute the OTOC for a maximally mixed initial state $\rho_0 = \Id/2^N$~\cite{huang2019finite}, as we here consider a pure initial state.
Instead, we provide numerical evidence that $F_H$ does not decay exponentially for translationally invariant local Hamiltonians with and without the time-reversal symmetry.
We especially consider time-reversal symmetric Hamiltonians as they are an important subclass of Hamiltonians widely considered for thermalization~\cite{kim2014testing}.

After creating a random $k$-local Hamiltonian $H$ (see Appendix~\ref{app:gen_rand_ham} for numerical details), we diagonalize $H$ and compute $F_H$ using the obtained eigenstates for the initial state $\ket{\psi} = \ket{+}^{\otimes N}$, $G=\sum_i Y_i$, and $O = \sum_i Z_i$.
As all observables, the Hamiltonian, and the initial state are translationally-invariant, we can compute $F_H$ within the translationally-invariant subspace.

We plot the result in Fig.~\ref{fig:random_ham_f} up to $N=16$ for random Hamiltonians without (main figure) and with (inset) the time-reversal symmetry. For $k=2$, we observe that the lower bound $F_H$ does not decay at all in both cases. For $k=3,4$, $F_H$ decreases with $N$ for the Hamiltonians without the time-reversal symmetry.
Even though it is not conclusive to tell the exact decaying rate of $F_H$ from this plot, we strongly believe that it is not exponential from the universality of thermalization dynamics for non-integrable models~\cite{deutsch1991quantum,srednicki1994chaos,rigol2008thermalization}, i.e., if $F_H$ for Hamiltonians with the time-reversal symmetry does not decay exponentially, $F_H$ for any thermalizing Hamiltonians also does not decay exponentially.

We also recall Ref.~\cite{larocca2022diagnosing}, which conjectured that the variance of gradients only scales inverse polynomially with the dimension of the dynamical Lie algebra $\mathcal{G}$ spanned by the gate generators, i.e., $\mathcal{G}=\braket{iG_1,\cdots,i G_l}_{\rm Lie}$ where $\braket{\cdot}_{\rm Lie}$ is the Lie closure containing all nested commutators of the listed elements.
A reason behind using dynamical Lie algebra is that the algebra generates any arbitrary circuit that the given PQC can express, i.e., there is a $g \in \mathcal{G}$ such that $U(\pmb{\theta})=e^{g}$ for any $\pmb{\theta}$.
However, for random Hamiltonians, as in our case, it is more natural to use a vector space of the Hamiltonians (which also generates a unitary operation but not a Lie algebra) instead.
As the dimension of $k$-local random Hamiltonians is $\Theta(N)$,
the conjecture with a slightly relaxed condition would also indicate that the gradient does not decay exponentially.

We explicitly write down our version of the conjecture, which is also supported by our numerical results, as follows:
\begin{conjecture} \label{conj:poly_decaying_random_local_ham}
Let $V$ be a vector space of local Hamiltonians. Then for any given initial state $\ket{\psi_0}$, $G \in V$, and a local operator $O$, we have
\begin{align}
    &-\int_{\nu \in V} d\nu \braket{\psi_0 | [G, e^{i\nu}Oe^{-i\nu}] | \psi_0}^2 = \frac{1}{\mathrm{poly}(N)}
\end{align}
where $d \nu$ is a proper measure for the vector space.
\end{conjecture}

\section{Approximating HVA to local Hamiltonian evolution} \label{sec:circuit_local_ham_approx}
In the previous section, we argued that a circuit given by the time evolution under a local Hamiltonian does not have barren plateaus.
In this section, we find a parameter condition for which the HVA given by Eq.~\eqref{eq:HVA} is well approximated by local Hamiltonian evolution.
We first interpret the HVA as a unitary operator generated by a time-dependent Hamiltonian.
Then, we utilize the Floquet-Magnus (FM) expansion to obtain an effective \textit{time-independent} Hamiltonian that describes the HVA within a small error.

We here consider each Hamiltonian $H^{(j)}$ satisfying the following conditions: (C1) $H^{(j)}$ is a sum of commuting Pauli strings, (C2) $H^{(j)}$ is $k$-local (each term acts on at most $k$ nearby qubits in a given lattice), and (C3) each Pauli string of $H^{(j)}$ uniquely supports a subset $X \subset \Lambda$ (e.g., $H^{(j)}$ cannot have terms $X_1X_2$ and $Z_1Z_2$ simultaneously).
In other words, we consider $H^{(j)}$ defined as
\begin{align}
    H^{(j)} = \sum_{|X| \leq k} h^{(j)}_X,
\end{align}
where the summation is over all subsets of sites $X\subseteq \Lambda$ whose length is $\leq k$, and $\Lambda = \{1,\cdots,N\}$ is a set of all sites.
In addition, $h^{(j)}_X$ is a single Pauli string (if there is a term whose support is $X$) or $0$ (otherwise), and $[h^{(j)}_X,h^{(j)}_Y] = 0$ for all $X,Y \subseteq \Lambda$. 
As we assume that $H^{(j)}$ is $k$-local, $h^{(j)}_X=0$ if $X$ contains any non-nearby sites (i.e., if there are $a, b \in X$ such that the distance between $a$ and $b$ is larger than $k$).

We also define parameters for the FM expansion. First, we define
\begin{align}
    H_{\rm max} &= \max_j \sum_{|X| \leq k} \Vert h^{(j)}_X \Vert \nonumber \\
    &= \max_j |\{X \subseteq \Lambda: h^{(j)}_{X} \neq 0 \}|, \label{eq:h0_def}
\end{align}
where we obtained the last equality using the fact that $\{h^{(j)}_X\}$ are commuting Pauli strings.
Thus, $H_{\rm max}$ is the maximum number of terms in $H^{(j)}$.
We also introduce a parameter $J$ that upper bounds the local interaction strength
\begin{align}
    \max_j \sum_{X: X \ni a} \Vert h^{(j)}_X \Vert \leq J, \qquad \forall a \in \Lambda.\label{eq:def_J}
\end{align}
As $\{h^{(m)}_X\}$ are Pauli strings, we can use
\begin{align}
    J = \max_{a \in \Lambda} \max_j | \{X: X \ni a \text{ and } h_X^{(j)} \neq 0 \} |.
\end{align}
From the locality of Hamiltonians $\{H^{(j)}\}$, 
we have $H_{\rm max} = \Theta(N)$ (see discussion below Proposition~\ref{prop:short_time_grad_var}),
and $J$ is upper bounded by the number of vertices whose $L^1$ distance to the origin is $\leq k$, which is a constant for a finite-dimensional lattice.

We further assume that all parameters $\{\theta_{i,j}\}$ in the HVA are larger than $0$.
Under this setup, the following Proposition shows that the subcircuits $U_R$ and $U_L$ of the HVA can be approximated by the time evolution under a few-body Hamiltonian when the sum of parameters is small.

\begin{proposition}\label{prop:effective_ham_fm}
For the HVA composed of $H^{(j)}$, given in Eq.~\ref{eq:HVA}, we consider a subcircuit $U_R = e^{-iH^{(j-1)}\theta_{i,j-1}} \cdots e^{-iH^{(1)} \theta_{1,1}}$.
We additionally assume that all $H^{(j)}$ satisfy the conditions (C1-C3) defined above.
Then, there is a Hamiltonian $H_R^{(n)}$ acting on at most $(n+1)k$-local sites such that
\begin{align}
    &\Bigl\Vert U_R - e^{-i H_R^{(n)} t_R} \Bigr\Vert \nonumber \\
    & \leq 6 H_{\rm max}  2^{-n_0} t_R + \frac{2 H_{\rm max} (2kJ)^{n+1}}{(n+2)^2}(n+1)! t_R^{n+2} \label{eq:fm_bound}
\end{align}
with $t_R = \theta_{1,1} + \cdots + \theta_{i,j-1}$ for all $n \leq n_0 = \lfloor 1/(32 k J t_R) \rfloor$.
Likewise, there is a Hamiltonian $H_L^{(n)}$ that approximates $U_L = e^{-iH^{(q)}\theta_{p,q}}\cdots e^{-iH^{(j)}\theta_{i,j}}$ with the same error but for $t_L = \theta_{i,j} + \cdots + \theta_{p,q}$.
\end{proposition}

We derive the bound and properties of $H_{R,L}^{(n)}$ in Appendix~\ref{app:approximating_hva}. Our derivation is based on the truncated FM expansion rigorously proven in Ref.~\cite{kuwahara2016floquet}.

The above bound tells us that $U_{R,L}$ can be approximated by local Hamiltonian evolution when $t_{R,L}$ are small.
For example, when $t_{R} = \mathcal{O}(1/N)$ and for a constant $n$,
the first term in the bound is exponentially small in $N$ and the second term is $\mathcal{O}(1/N^{n+1})$.
Then, one may further employ Proposition~\ref{prop:short_time_grad_var} to get a large gradient, which we summarize as the following theorem:

\begin{theorem} \label{thm:small_parameter_large_grad}
For the HVA [Eq.~\eqref{eq:HVA}] and for a local observable $O$ acting on at most $\mathcal{O}(1)$ sites,
assume that there exists an initial state $\rho_0 = \ket{\psi_0}\bra{\psi_0}$ which gives $g:=|\Tr[\rho_0 [H^{(m)}, O]]| = \Theta(1)$ regardless of $N$.
Then, there is $\tau_0 = \Theta(1/N)$ such that
\begin{align}
    \Bigl| \frac{\partial C}{\partial \theta_{n,m}} \Bigr| \geq \frac{g}{4}
\end{align}
for all $n$, if $\sum_{ij}\theta_{i,j} = t_L + t_R \leq \tau_0$.
\end{theorem}
The proof can be found in Appendix~\ref{app:prove_small_parameter_large_grad}.

We provide two remarks on Theorem~\ref{thm:small_parameter_large_grad}.
First, if there exists any constraint with $\tilde{\tau}_0$ such that a gradient component is bounded below by a constant for all $\sum_{ij}\theta_{ij} \leq \tilde{\tau}_0$, then $\tilde{\tau}_0 \leq \pi/4$.
This is because one can easily find the HVA with suitable $\rho_0$ and $O$ which satisfies $|\Tr[\rho_0 [H^{(m)}, O]]| = \Theta(1)$, but $\partial_{n,m}C = i\Tr[U(\pmb{\theta}) \rho U(\pmb{\theta})^{\dagger}[H^{(m)},O]]=0$ for $\sum_{i,j}\theta_{i,j}=\pi/4$. We provide such an example in Appendix~\ref{app:finite_time_evolution_zero_grad}.
This implies that Theorem~\ref{thm:small_parameter_large_grad} can be improved at most $\tau_0 = \Theta(1)$ in our current set-up.

Second, the theorem implies that for \textit{any} probability distribution $p(\pmb{\theta})$ defined for $\theta_{i,j} \geq 0$ and $\sum_{i,j}\theta_{i,j} \leq \tau_0$, 
\begin{align}
    \int d\pmb{\theta} p(\pmb{\theta}) \Bigl( \frac{\partial C}{\partial \theta_{n,m}} \Bigr)^2 = \Theta(1),
\end{align}
which is much stronger than the non-exponential decay of gradients.
If one considers a different condition, e.g., a polynomially decaying gradients under uniformly generated initial parameters,
a parameter constraint may be further relaxed. 
Namely, we open up the possibility that a constraint with $\tau_1 = \Omega(1)$ exists such that
\begin{align}
    \int_{\theta_{i,j} \geq 0, \sum_{i,j}\theta_{i,j} \leq \tau_1} \prod_{i,j} d\theta_{i,j} \Bigl( \frac{\partial C}{\partial \theta_{n,m}} \Bigr)^2 =  \frac{1}{\mathrm{poly}(N)}
\end{align}
is satisfied for all $\theta_{i,j}\geq 0$ and $\sum_{i,j}\theta_{i,j} \leq \tau_1$.
We numerically investigate related scenarios in the following section.

\section{Numerical comparison between initialization methods} \label{sec:numerical_comparisions_init}
Theorem~\ref{thm:small_parameter_large_grad} tells us that there is $\tau_0=\Theta(1/N)$ such that the HVA does not have barren plateaus when the sum of all parameters is less than $\tau_0$.
Still, the exact value of $\tau_0$ for the Theorem is difficult to obtain or can be unrealistically small.
Thus, in this subsection, we introduce an initialization method based on Theorem~\ref{thm:small_parameter_large_grad} and compare it to the small constant initialization considered in Refs.~\cite{wierichs2020avoiding,park2021efficient,bosse2022probing,kattemolle2022variational}.

We numerically test the following three different initialization methods.
(1) \textbf{Random}: complete uniformly random initialization, such that all parameters are from $\mathcal{U}_{[0, 2\pi]}$,
(2) \textbf{Constrained}: the sum of parameters in each layer is constrained to be $T = c/N$ with a constant $c$, i.e., $\sum_{j}\theta_{i,j}=T$ for all $i$, and
(3) \textbf{Small}: $\theta_{i,j} \sim \mathcal{U}_{[0, \epsilon]}$ for a small $\epsilon$ independent to $N$~\cite{wierichs2020avoiding,park2021efficient,bosse2022probing}.
For method (2), we show that a relatively large value of $c=\pi/2$ already gives $\Theta(1)$ gradient magnitudes.
On the other hand, we observe two different scaling behaviors for the constant small initialization [method (3)].
There is a value $N_0$ depending on $p$ and $\epsilon$ such that the gradient magnitudes decay exponentially for $N < N_0$ whereas they only decay polynomially for $N > N_0$.
This observation suggests that there can be another parameter regime in the HVA that does not have barren plateaus.

We use the HVA for the one-dimensional (1D) and two-dimensional (2D) spin-$1/2$ Heisenberg-XYZ models $\mathcal{H} = \sum_{\braket{a,b}} J_x X_a X_b + J_y Y_a Y_b + J_z Z_a Z_b$ to test these methods, which is given by
\begin{align*}
    \ket{\psi(\pmb{\theta})} &=  \prod_{i=p}^1 e^{-i \theta_{i,3} \sum_{\braket{a,b}} Z_a Z_b } e^{-i \theta_{i,2} \sum_{\braket{a,b}} Y_a Y_b } \\
    & \qquad \times e^{-i \theta_{i,1} \sum_{\braket{a,b}} X_a X_b } \ket{\psi_0} \numberthis \label{eq:hva_xyz}
\end{align*}
where $\braket{a,b}$ are two nearest neighbors and $\ket{\psi_0}$ is the N\'{e}el state in the given lattice.
We use $N$ spins in the periodic boundary condition (a ring) for the one-dimensional model.
For the two-dimensional model, we consider a rectangular lattice with the periodic boundary condition (a torus) size of $L_x \times L_y$.
Thus the N\'{e}el state in these lattices are given by $\ket{\psi_0} = (\ket{\downarrow \uparrow}^{\otimes N/2} + \ket{\uparrow \downarrow }^{\otimes N/2})/\sqrt{2}$ and $[(\ket{\downarrow \uparrow}^{\otimes L_x/2}\ket{\uparrow \downarrow}^{\otimes L_x/2})^{\otimes L_y/2}+ (\ket{\uparrow \downarrow}^{\otimes L_x/2}\ket{\downarrow \uparrow}^{\otimes L_x/2})^{\otimes L_y/2}]/\sqrt{2}$ for the 1D and 2D models, respectively.

\begin{figure}
    \centering
    \includegraphics[width=0.9\linewidth]{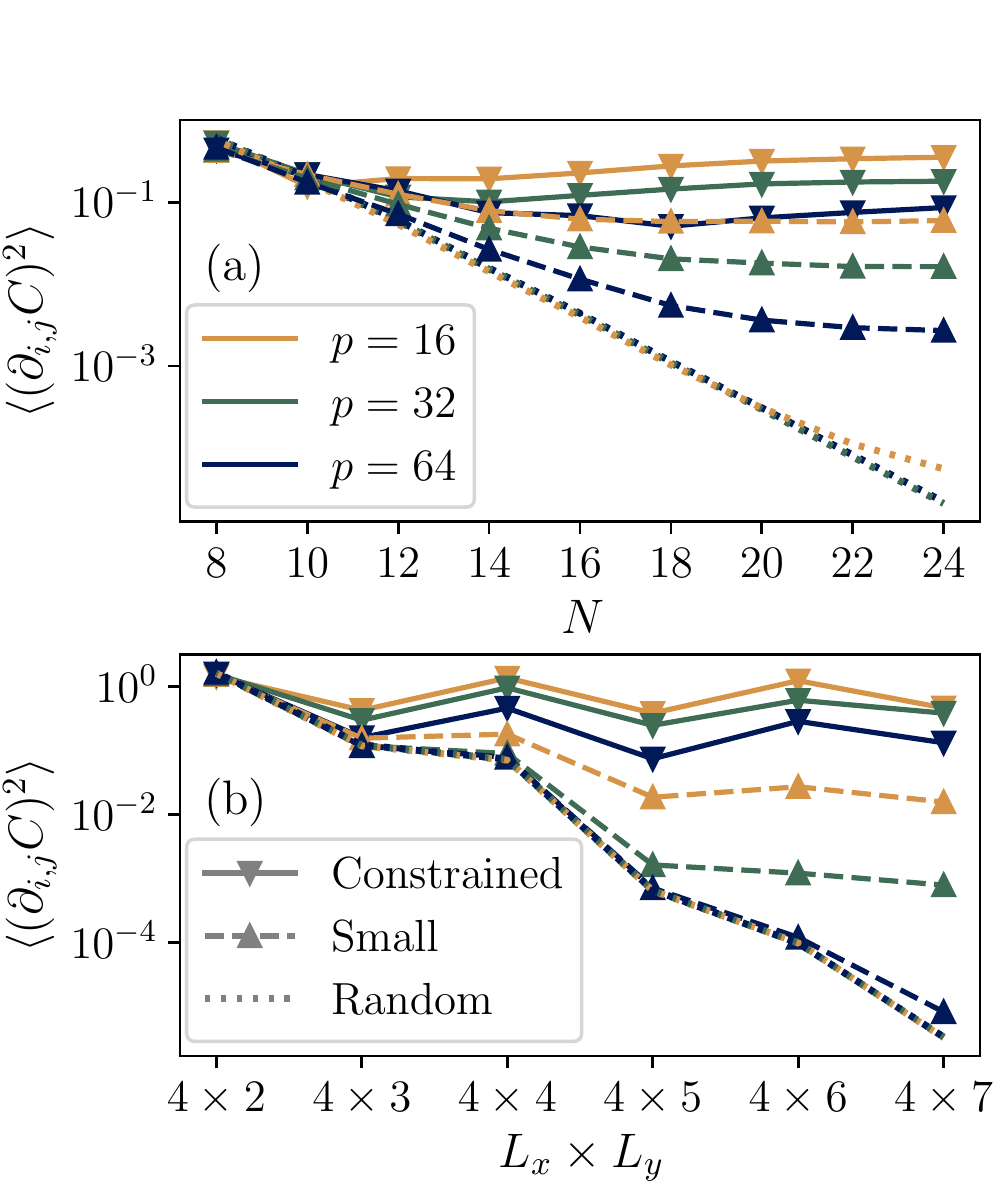}
    \caption{(a) Scaling of the gradient square $(\partial_{i,j}C)^2$ from the HVA for the 1D Heisenberg-XYZ model with different depths $p$ (see main text for details).
    Plots show results from the constrained (solid), small (dashed), and completely random (dotted) parameter initializations.
    We compute $(\partial_{i,j}C)^2$ for $2^{10}$ parameters samples from each distribution and plot the averaged results over all samples and gradient components ($i,j$).
    (b) The same plot as (a) but for the 2D Heisenberg-XYZ model with the lattice size $L_x \times L_y$.
    We also see that the results fluctuate for odd and even $L_y$ because the 2D N\'{e}el state (our initial state) violates the symmetry of the Hamiltonian for odd $L_y$. 
    We also compute the relative standard deviation, $r=\sigma(X)/\mathbb{E}[X]$, where $X:=\sum_{ij} (\partial_{i,j}C)^2/(3p)$ is the squared partial derivatives averaged over all parameters. Here, the standard deviation and the expectation value are taken over the circuit instances.
    The values of $r$ for the 1D model with $p=64$ and $N=24$ are given by $0.20$, $0.53$, and $0.13$ for the contained, small, and random initialization, respectively.
    For the 2D model with $p=64$ and $L_x \times L_y = 4 \times 7$, we obtained $r\approx 0.39$ (constrained), $1.02$ (small), $0.11$ (random), respectively. 
    }
    \label{fig:grad_scaling}
\end{figure}

\subsection{Scaling of gradients}
We compute gradients of the cost function with $O=Y_0Y_1$ (thus $C=\braket{\psi(\pmb{\theta})|Y_0Y_1|\psi(\pmb{\theta})}$) obtained from different initialization methods.
For constrained initialization, random values $\tilde{\theta}_{i,j}$ are first sampled from the uniform distribution, i.e., $\tilde{\theta}_{i,j} \sim \mathcal{U}_{[0, 2\pi]}$, and then parameters are assigned by normalizing them: $\theta_{i,j}= \tilde{\theta}_{i,j} \times T/(\sum_{j=1}^3 \tilde{\theta}_{i,j})$ for all $i,j$. This method ensures that $\sum_{j=1}^3 \theta_{i,j}=T$. We here use $T=\pi/(2N)$.
The results are compared to small parameters $\theta_{i,j} \sim \mathcal{U}_{[0, \epsilon]}$ with $\epsilon=0.2$ as well as complete random parameters $\theta_{i,j} \sim \mathcal{U}_{[0,2\pi]}$.
For each set of system parameters (size and depth $p$) and the initialization method,
we compute all gradient components and plot the averaged squared magnitudes (i.e., we averaged $\langle (\partial_{i,j} C)^2\rangle = \sum_{i,j} (\partial_{i,j} C)^2/(3p)$ over $2^{10}$ random circuit instances).

The results for the 1D and 2D models are shown in Fig.~\ref{fig:grad_scaling}(a) and (b), respectively.
The 1D model clearly shows that the magnitudes of gradients do not decay with $N$, i.e., $(\partial_{i,j}C)^2 \approx \Theta(1)$, for the constrained initialization, which is consistent with Theorem~\ref{thm:small_parameter_large_grad}.
On the other hand, the gradient magnitudes decay exponentially with $N$ when the complete random initialization is used with $p \in [16, 32]$.
However, we could observe an interesting behavior when parameters are initialized to be small ($\theta_{i,j}\sim\mathcal{U}_{[0,\epsilon]}$).
In this case, the gradient decays exponentially up to some $N_0$, i.e., for $N \leq N_0 \approx 18$, but it decays slower after that.
One can already see this signature even for the complete random initialization when $p=16$ and $N=24$, where the averaged gradient magnitudes are larger than that from $p\in[32,64]$.
We also see similar behaviors for the 2D model, besides the results from each initialization oscillate for odd and even $L_y$, since our initial state (2D N\'{e}el state) is not fully symmetric for odd $L_y$.

\begin{figure}
    \centering
    \includegraphics[width=0.9\linewidth]{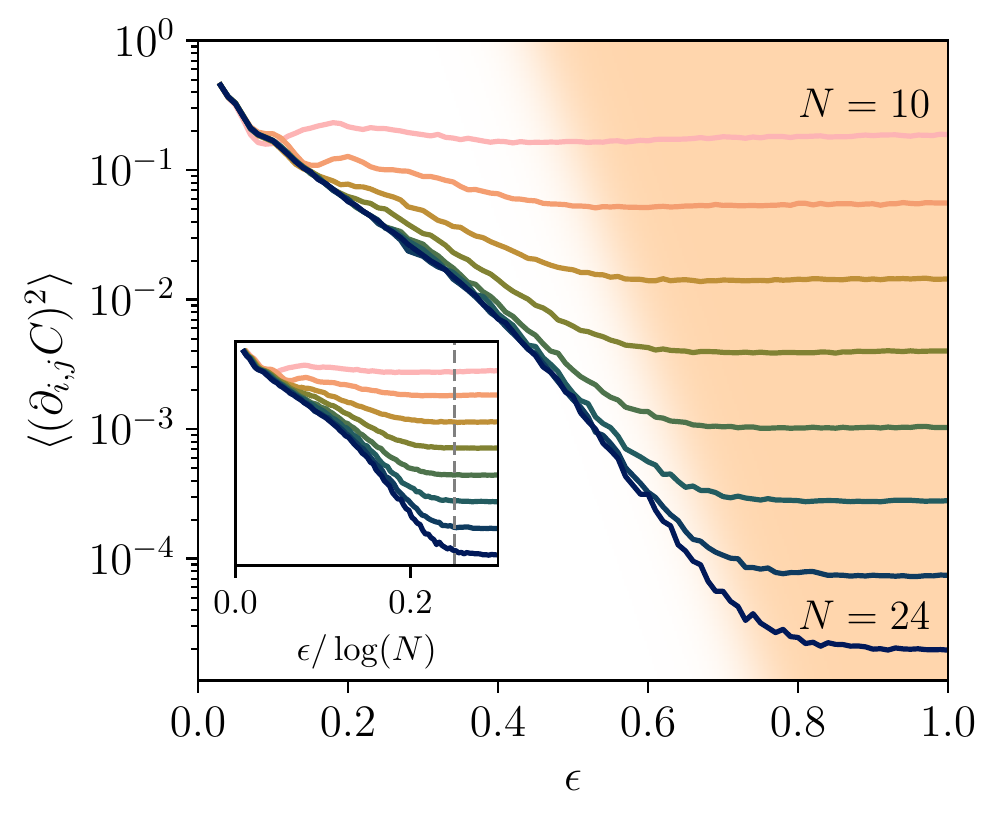}
    \caption{Averaged magnitudes of gradients $\langle (\partial_{i,j} C)^2 \rangle $ as a function of $\epsilon$. The HVA for the 1D XYZ model with $p=16$ is used. All parameters are samples from the uniform distribution, i.e., $\theta_{i,j} \sim \mathcal{U}_{[0, \epsilon]}$. We observe that there is a value $\epsilon_0(N)$ such that the gradient decays exponentially with $N$ only for $\epsilon > \epsilon_0(N)$ (colored region).
    Inset shows the same data but as a function of $\epsilon/\log(N)$.
    We see the gradient magnitudes saturate after the dashed vertical line, which suggests that $\epsilon_0(N) \propto \log N$.}
    \label{fig:gradient_scaling_eps}
\end{figure}

As non-exponential decaying gradients from the small constant initialization have not been clearly reported in previous studies, we explore these phenomena more closely here.
We compute the averaged squared gradients using the HVA for the 1D XYZ model with $p=16$ when the circuit parameters are sampled from $\mathcal{U}_{[0,\epsilon]}$ for different values of $N$ and $\epsilon$.
We plot the result as a function of $\epsilon$ in Fig.~\ref{fig:gradient_scaling_eps}.
The results show that the magnitudes of gradients saturate to an exponentially small value as $\epsilon$ increases, but the point it saturates, $\epsilon_0(N)$, also increases with $N$. For $\epsilon < \epsilon_0(N)$, we observe that the gradient does not decay exponentially.
This fact also confirms that there can be another parameter regime beyond the one we mainly considered in this paper, which is also free from barren plateaus.

To see how $\epsilon_0(N)$ scales with $N$, we plot the same data but as a function of $\epsilon/\log(N)$ (inset).
The plot shows that the gradient magnitudes saturate when $\epsilon/\log(N)$ is larger than a constant, which suggests that $\epsilon_0(N) \propto \log N$.
In general, we expect that there is a relation between $\Upsilon := \sum_{i,j} \theta_{i,j}$ (which is $\propto p\epsilon$ in this case) in the HVA for a $2$-local Hamiltonian and a random local circuit with depth $\propto \Upsilon$.
Such a connection explains the observed behavior as a 1D random local circuit requires its depth larger than $\Theta(\log(N))$ to show exponential decay of gradient magnitudes when the cost function is given by the expectation value of a local observable~\cite{cerezo2021cost}.

We still note that it is less clear whether a small constant parameter initialization [method (3)] gives the same quantitative behavior for the HVAs with more complex Hamiltonians (e.g., 1D $k$-local with $k \geq 3$ or defined in a high-dimensional lattice).
In contrast, we expect to have $\Theta(1)$ gradient magnitudes regardless of the dimension with our initialization method [method (2)].
As a detailed investigation of the relation between the HVA and a local random circuit is out of the scope of the current work, we leave it to future work.

\begin{figure}
    \includegraphics[width=0.95\linewidth]{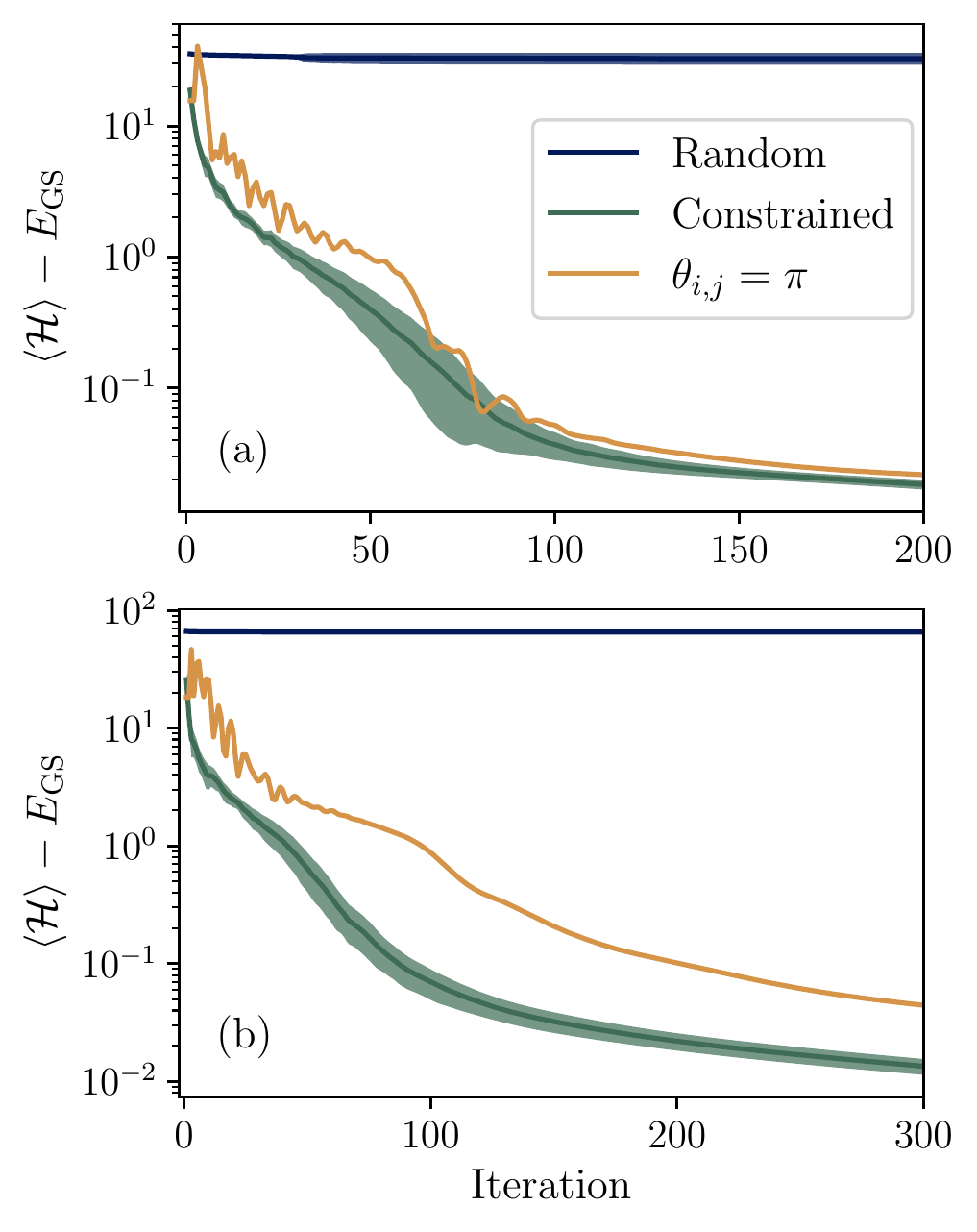}
    \caption{
    (a) Learning curves from different initialization schemes for the 1D Heisenberg model with $N=20$. We use the circuit ansatz with $p=20$ and the Adam optimizer with the learning rate $\alpha=0.025$. For each iteration, the curves for random and constrained initializations show the averaged results over $32$ different initial parameters. The shaded regions indicate one standard deviation ($[m-\sigma/2,m+\sigma/2]$). Note that results from the constant initialization ($\theta_{i,j}=\pi$) do not vary between instances as there is no randomness in the simulation.
    The ground state energy $E_{\rm GS}$ is obtained using the exact diagonalization.
    (b) The same result for the 2D Heisenberg model with $L_x \times L_y = 4 \times 6$ lattice. The circuit ansatz with $p=24$ and the Adam optimizer with learning rate $\alpha=0.005$ are used. Results for random and constrained initializations are averaged over $16$ different initial parameters.
    For optimization, we compute the gradient exactly (without shot noise).
    Shaded regions are barely visible for the random initialization.
    This is because the loss function, $\braket{\cal H}$, is not trained at all for most instances.
    Thus, the loss function preserves its initial value $\braket{\cal H} \approx 0$, which is from the fact that the circuit forms a 2-design for the random initialization.
    }
    \label{fig:vqe_heisenberg}
\end{figure}

\subsection{Full simulation of variation quantum eigensolver}
We now explore whether our initialization improves learning procedures by fully simulating a variational quantum eigensolver (VQE) using the Heisenberg model ($J_x = J_y = J_z = 1$).
We define the Hamiltonian expectation values as the cost function ($C = \braket{\psi(\pmb{\theta})|\mathcal{H}|\psi(\pmb{\theta})}$) 
and train the circuit using the Adam optimizer~\cite{kingma2014adam}.

We first simulate the VQEs using exact gradients. 
Quantum hardware cannot compute exact gradients, as each gradient component should be estimated from the measurement outcomes from shots.
However, classical quantum circuit simulators support multiple algorithms to obtain exact gradients.
For our simulation, we use the adjoint method~\cite{jones2020efficient} implemented in PennyLane~\cite{bergholm2018pennylane}.
We present learning curves from different parameter initialization methods for the one-dimensional lattices (with learning rate $\alpha = 0.025$ and the default values for hyperparameters $\beta_1=0.9$ and $\beta_2=0.999$) in Fig.~\ref{fig:vqe_heisenberg} (a), which shows that our initialization scheme outperforms other initialization schemes.
The completely random parameter initialization fails to find the ground state. This is an expected behavior from the presence of the barren plateaus.
On the other hand, initializing all parameters to $\pi$ ($\theta_{i,j} = \pi$ for all $i,j$), which is used in Ref.~\cite{wiersema2020exploring}, works but is subject to large initial fluctuations.
Generally speaking, such an initialization without randomness is prone to local minima~\cite{wessels1992avoiding}.
We also found a similar behavior for the 2D Heisenberg model, shown in Fig.~\ref{fig:vqe_heisenberg}(b).

\begin{figure}[t]
    \centering
    \includegraphics[width=0.95\linewidth]{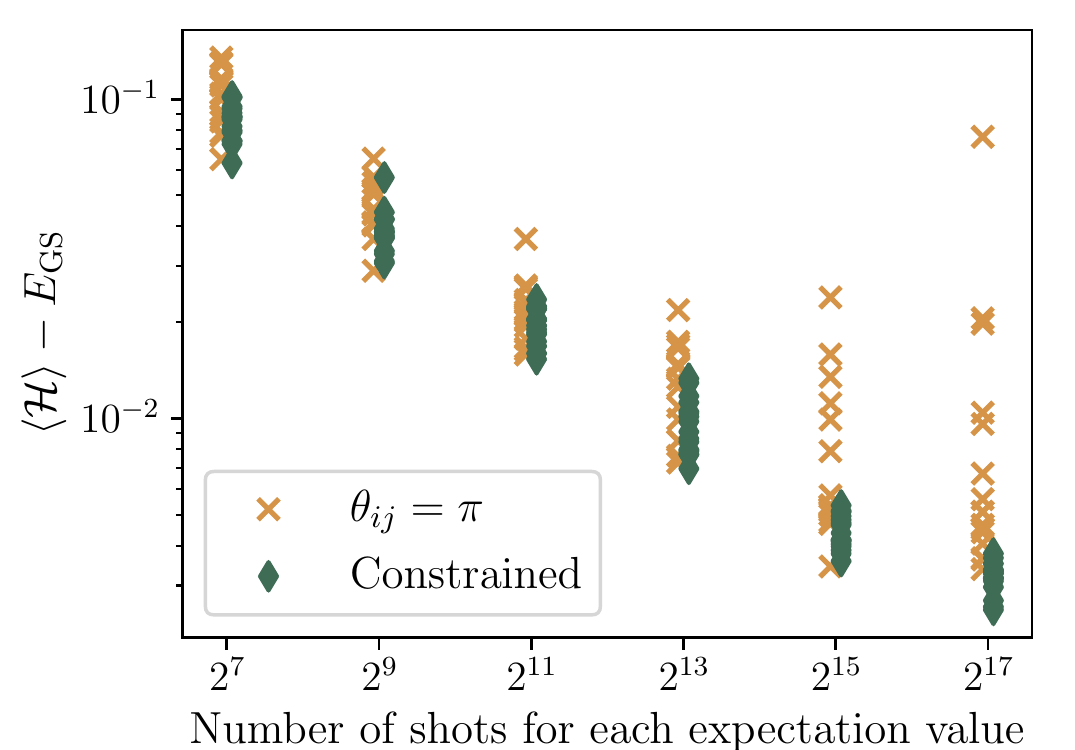}
    \caption{
        Converged energies from the VQE for the one-dimensional Heisenberg model with $N=p=16$ as a function of the number of shots. For each number of shots, $n_{\rm shot} \in [2^7, 2^9, 2^{11}, 2^{13}, 2^{15}, 2^{17}]$, we fully simulate the VQE $16$ times.
        The converged energies $\braket{\mathcal{H}}$ for each independent VQE instance are presented.
        For the initialization $\theta_{ij}=\pi$, the shot noise is the only source of the randomness.
        On the other hand, the initial parameters are also random when the constrained parameter initialization is used.
    }
    \label{fig:hei_1d_shots_converged}
\end{figure}

We next compare results from different initialization methods when a finite number of shots is used.
We solve the 1D Heisenberg model, but gradients are now estimated using $n_{\rm shot}$ shots.

The gradient with a finite number of shots can be obtained as follows. First, we introduce another PQC with the same shape as the HVA, Eq.~\eqref{eq:hva_xyz}, but all gates have different parameters. Such a PQC is written as
\begin{align*}
    &\ket{\psi(\pmb{\alpha}, \pmb{\beta}, \pmb{\gamma})} \\
    &=  \prod_{i=p}^1 e^{-i \sum_{\braket{a,b}}\gamma_{i,a,b} Z_a Z_b } e^{-i \sum_{\braket{a,b}} \beta_{i,a,b} Y_a Y_b } \\
    & \qquad \times e^{-i \sum_{\braket{a,b}} \alpha_{i,a,b} X_a X_b } \ket{\psi_0}. \numberthis \label{eq:hva_xyz_expand_params}
\end{align*}
Compared to Eq.~\eqref{eq:hva_xyz}, all gates now have independent parameters.
Next, we obtain each gradient component using the two-term parameter-shift rule~\cite{mitarai2018quantum,schuld2019evaluating}.
By defining $f(\pmb{\alpha}, \pmb{\beta}, \pmb{\gamma}) = \braket{\psi(\pmb{\alpha}, \pmb{\beta}, \pmb{\gamma})|\mathcal{H}|\psi(\pmb{\alpha}, \pmb{\beta}, \pmb{\gamma})}$, we obtain its gradient for $\alpha_{i,a,b}$ as follows:
\begin{align}
    \frac{\partial f }{\partial \alpha_{i,a,b}} &=  \frac{1}{2} \Bigl[ f\bigl(\pmb{\alpha} + \frac{\pi}{2} \pmb{\delta}_{i,a,b}, \pmb{\beta}, \pmb{\gamma} \bigr) \nonumber \\
    & \qquad - f\bigl(\pmb{\alpha} - \frac{\pi}{2}\pmb{\delta}_{i,a,b}, \pmb{\beta}, \pmb{\gamma} \bigr)\Bigr],
\end{align}
where $\pmb{\delta}_{i,a,b}$ is a vector components of which is $1$, if the index is $(i,a,b)$, or $0$, otherwise.
Gradient for $\beta$ and $\gamma$ also can be obtained similarly.

Shot noise is introduced when we estimate $f(\pmb{\alpha}, \pmb{\beta}, \pmb{\gamma})$. From
\begin{align}
    \braket{\mathcal{H}} = \sum_{\braket{a,b}} \braket{X_a X_b} + \braket{Y_a Y_b} + \braket{Z_a Z_b},
\end{align}
we can estimate $f=\braket{\psi(\pmb{\alpha}, \pmb{\beta}, \pmb{\gamma})|\mathcal{H}|\psi(\pmb{\alpha}, \pmb{\beta}, \pmb{\gamma})}$ using the samples of $\psi(\pmb{\alpha}, \pmb{\beta}, \pmb{\gamma})$ in the $X$, $Y$, and $Z$ bases.
For example, let us estimate $\braket{X_a X_b}$ using $n_{\rm shot}$ samples.
We first obtain bitstrings $\{x^{(1)}, \cdots, x^{(n_{\rm shot})}\}$ from the probability distribution $p(x)=|\braket{x|H^{\otimes N} | \psi(\pmb{\alpha}, \pmb{\beta}, \pmb{\gamma})}|^2$, where each $x^{(i)} \in \{0, 1\}^N$ is a bitstring with length $N$.
Then, each $\braket{X_aX_b}$ can be estimated using these samples by 
\begin{align}
    \braket{X_aX_b} \approx \frac{1}{n_{\rm shot}} \sum_{i=1}^{n_{\rm shot}} (1- 2 x^{(i)}_a) (1- x^{(i)}_b),
\end{align}
where $1-2x^{(i)}_a$ is from the fact that $X_a$ has the value $1$ if $x^{(i)}_a =0$, and $-1$ if $x^{(i)}_a =1$.
Note that we can use the same set of samples, $\{x^{(1)}, \cdots, x^{(n_{\rm shot})}\}$, for all $\braket{X_a X_b}$.
Thus, $f = \braket{\mathcal{H}}$ for each set of parameters is estimated using $3 n_{\rm shot}$ samples, and each gradient component is estimated using $6 n_{\rm shot}$ samples.

Finally, the gradient of the original HVA, Eq.~\eqref{eq:hva_xyz}, which shares the parameters between the gates, can be obtained by summing over the gradient components.
Namely, we have
\begin{align}
    \frac{\partial \braket{\psi(\pmb{\theta})|\mathcal{H}|\psi(\pmb{\theta})}}{\theta_{i,1}} = \sum_{\braket{a,b}} \frac{\partial f}{\partial \alpha_{i,a,b}},\\
    \frac{\partial \braket{\psi(\pmb{\theta})|\mathcal{H}|\psi(\pmb{\theta})}}{\theta_{i,2}} = \sum_{\braket{a,b}} \frac{\partial f}{\partial \beta_{i,a,b}},\\
    \frac{\partial \braket{\psi(\pmb{\theta})|\mathcal{H}|\psi(\pmb{\theta})}}{\theta_{i,3}} = \sum_{\braket{a,b}} \frac{\partial f}{\partial \alpha_{i,a,b}}.
\end{align}
Given that the circuit has $3Np$ gates in total, and each gradient component is estimated from $6 n_{\rm shot}$ samples, 
$18 N p n_{\rm shot}$ samples are used for each iteration to estimate all gradient components.

For the HVA with $N=p=16$ [see Eq.~\eqref{eq:hva_xyz}], we plot the converged energies after $10^3$ iterations from the VQE simulations as a function of $n_{\rm shot}  \in [2^7, 2^9, 2^{11}, 2^{13}, 2^{15}, 2^{17}]$ in Fig.~\ref{fig:hei_1d_shots_converged}.
While the results show that the best-converged energies from our constrained initialization scheme and $\theta_{ij}=\pi$ are similar, the deviations between instances from our initialization scheme are much smaller.
For example, the worst-performing instance for $n_{\rm shot}=2^{15}$ gives $\braket{\mathcal{H}} - E_{\rm GS} \approx 3 \times 10^{-2}$ when all parameters are initialized with $\pi$, but that from our initialization is $\approx 6 \times 10^{-3}$.

Our parameter constraint can also be imposed throughout the optimization steps (the ansatz itself).
When imposed on the ansatz, we can slightly change the cost function to ensure the parameters always follow the constraints. 
This can be simply done by assigning $\theta_{j,q} = T - \sum_{i=1}^{q-1} \theta_{j,i}$ and replacing the cost function $C$ with $\tilde{C}=[\prod_{i,j} \mathbf{1}(\theta_{j,i})] [\prod_j \mathbf{1}(T - \sum_{i=1}^{q-1} \theta_{j,i})] C$ where $\mathbf{1}(x)$ is the Heaviside step function. 
One can see this (piecewise differentiable) cost function (1) restricts all parameters to be larger than $0$, and (2) the sum of the parameters in each block is given by $T$ throughout the training.

Still, the constraint-imposed ansatz may not be useful for solving complex problems, as we require $pT$ to be small.
Even though the ansatz itself allows a large-depth circuit (e.g., $p = \Theta(N^2)$ and $T =1/N^3$), the Suzuki-Trotter decomposition~\cite{suzuki1991general} tells us that such a circuit always can be approximated by a short-depth circuit, i.e., there is another circuit with depth $d = \mathrm{poly}(pT) = \mathrm{poly}(1/N)$ that can express our constrained HVA with a small error. 
Using the notion of quantum circuit complexity~\cite{nielsen2005geometric,nielsen2006quantum,stanford2014complexity,haferkamp2022linear}, defined by the minimum number of two-qubit gates in any circuit that implements the given unitary,
we can say that this ansatz has a small approximate circuit complexity.

\begin{figure}[t]
    \centering
    \includegraphics[width=0.95\linewidth]{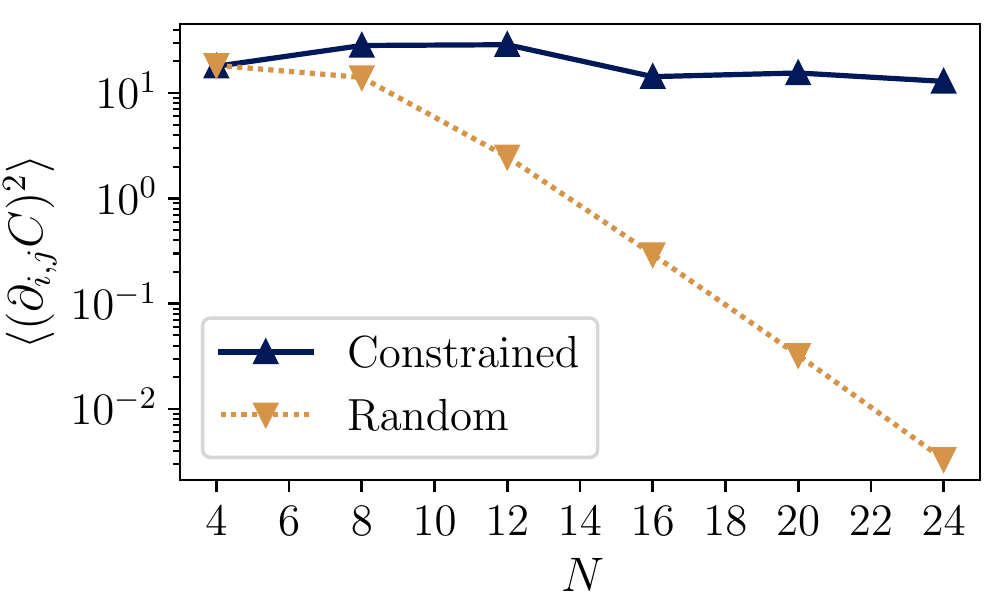}
    \caption{
    Scaling of the averaged squared gradients $(\partial_{i,j}C)^2$ from the repeated ansatz Eq.~\eqref{eq:repeat_ansatz} with (solid) and without (dotted) a parameter constraint.
    For the parameter-constrained ansatz, we sample parameters under the constraint $\theta_{i,1}+\theta_{i,2}+\theta_{i,3}=\pi/(2N)$.
    In contrast, $\theta_{i,j} \sim \mathcal{U}_{[0,2\pi]}$ is used for the ansatz without the constraint.
    The HVA is for the 1D XYZ model with $O=Y_0Y_1$ with $\tilde{p}=16$ and $r=N^2/4$. The results are averaged over $2^{10}$ random parameters.
    }
    \label{fig:hva_grad_repeat}
\end{figure}

\subsection{Long-time evolution with repeated parameters}
To overcome the problem that a simple parameter-constrained ansatz introduced in the previous subsection is not expressive enough, we propose another ansatz with better expressivity.
Our solution is to repeat the circuit multiple times instead of adding free parameters. In this case, the circuit is given as
\begin{align}
    U(\pmb{\theta}) = \Bigl[ \prod_{i=\tilde{p}}^1 e^{-iH^{(q)} \theta_{i,q}} \cdots e^{-iH^{(1)} \theta_{i,1}} \Bigr]^r \label{eq:repeat_ansatz}
\end{align}
with the constraint $\sum_j \theta_{i,j} = T$.
Thus, the circuit has a total of $\tilde{p}qr$ layers but only has $\tilde{p}q$ parameters.
This ansatz can be approximated by $e^{-iK (\tilde{p}rT)}$ for a local Hamiltonian $K$ with an error $\mathcal{O}(r(\tilde{p}T)^{n+2})$ when $\tilde{p}T$ is inverse polynomial with $N$ (i.e., $\tilde{p}T = \mathcal{O}(N^{-\gamma})$ for $\gamma > 0$).
This fact follows from Proposition~\ref{prop:effective_ham_fm} and $\Vert U_1^r - U_2^r \Vert \leq r \Vert U_1 - U_2 \Vert$ which holds for arbitrary $U_1$ and $U_2$~\footnote{This is from $\Vert U_1^r - U_2^r \Vert = \Vert U_1 U_1^{r-1} - U_1 U_2^{r-1} + U_1 U_2^{r-1} - U_2 U_2^{r-1} \Vert \leq \Vert U_1^{r-1} - U_2^{r-1} \Vert + \Vert U_1 - U_2 \Vert $, where we used $\Vert U_1\Vert = \Vert U_2^{r-1}\Vert = 1$.}.

We then further expect that the gradients scale polynomially with $N$ from Conjecture~\ref{conj:poly_decaying_random_local_ham} when $\tilde{p}rT$ is sufficiently large enough to equilibrate the system~\footnote{Precisely, one also require an ergodicity assumption that resulting Hamiltonians ($K$s) are uniformly distributed over the vector space for random parameters $\theta_{i,j}$.}.
We also numerically test the gradient scaling of this ansatz for $r=N^2/4$, $\tilde{p}=16$, and $T=\pi/(2N)$ using the HVA for the one-dimensional XYZ model in Fig.~\ref{fig:hva_grad_repeat}.
The plot shows that the gradient does not decay when the parameters are constrained, whereas it decays exponentially otherwise.

We now argue that the repeated ansatz of Eq.~\eqref{eq:repeat_ansatz} (1) can generate sufficiently complex unitary operators and (2) is useful for variational time evolution~\cite{li2017efficient,yuan2019theory,cirstoiu2020variational,lin2021real}.
The complexity of the circuit directly follows from the observation that the circuit approximates to $e^{-iK(\tilde{p}rT)}$ for a Hamiltonian $K$ with a large $\tilde{p}rT$.
It is commonly believed that simulating the long-time evolution of a general local Hamiltonian requires a large depth circuit (which is also formally conjectured in Refs.~\cite{brown2017quantum,brown2018second} in terms of quantum circuit complexity). 
Next, the given ansatz can express the time evolution of a given Hamiltonian $H = \sum_{j=1}^{\tilde{p}} \alpha_j H_j$.
Using the first-order Suzuki-Trotter decomposition, we can write 
\begin{align}
    e^{-iHt} = (e^{-iH t/r})^{r} \approx \Bigl[ \prod_{j=\tilde{p}}^1 e^{-i\alpha_j H_j t_0} \Bigr]^{r}
\end{align}
with $t_0 = t/r$ and an error $\mathcal{O}(N r t_0^2)$.
Thus the approximation has an error $\mathcal{O}(1/N)$ if we use $t_0 = \Theta(1/N)$.
As the right-hand side is nothing but Eq.~\eqref{eq:repeat_ansatz} with $T = \sum_j \alpha_j t_0$,
the ansatz with $T = \Theta(1/N)$ can approximate $e^{-iHt}$.

\section{Conclusion}
\label{sec:conclusion}
We studied the scaling behaviors of the gradients in the hamiltonian variational ansatz (HVA) and showed that adding a simple parameter constraint to the ansatz results in large gradients.
We demonstrated that the gradient magnitudes scale as $\Theta(1)$ when the circuit is given by short-time evolution and $1/\mathrm{poly}(N)$ when it is given by long-time evolution.
For the short-time regime, we provided a rigorous proof based on the rate of the gradient evolution, while we showed numerical evidence based on quantum thermalization~\cite{deutsch1991quantum,srednicki1994chaos,rigol2008thermalization,huang2019finite} for long-time evolution.
We then found the parameter constraints for which the HVA can be approximated by short-time as well as long-time evolution under a local Hamiltonian.
We further supported our arguments with extensive numerics for up to $28$ qubits, which also consistently showed the correctness of our arguments.

For long-time evolution, our argument is based on the fact that the dynamics generated by thermalizing Hamiltonians are more restricted than unitary 2-designs.
Albeit typical Hamiltonians thermalize~\cite{gogolin2016equilibration},
there are two other important classes of Hamiltonians with different dynamic properties: integrable and many-body localized systems.
In contrast to thermalizing systems where information of initial states spreads out through the Hilbert space (but within a subspace preserving the energy),
initial information on integrable and many-body localized systems can be easily accessed by simple operators at any time, i.e., their dynamics are even more restrictive than thermalizing Hamiltonians.
Given this interesting property, we expect that there could be a different parameter condition that parameterized quantum circuits approximate to integrable or many-body localized systems, which are also free from barren plateaus.

For example, it is known that the out-of-time correlator of many-body localized systems does not decay exponentially with the system size for particular choices of observables and initial states (see, e.g., Refs.~\cite{chen2016universal,fan2017out,lee2019typical}).
Following the arguments in Sec.~\ref{sec:grad_scaling_long_time}, we hope to find a class of parameterized quantum circuits that approximate to a many-body localized system and have large gradients.
On the other hand, Ref.~\cite{larocca2022diagnosing} showed that the dynamic Lie algebra $\mathcal{G}$ generated by the HVA for the XXZ model ($J_x = J_y$) can have a small dimension, i.e., the Lie algebra $\langle i\sum_i(X_i X_{i+1} + Y_i Y_{i+1}), i\sum_i Z_i Z_{i+1}\rangle_{\rm Lie}$, has a small dimension.
As the XXZ model is solvable by the Bethe ansatz (thus integrable), we believe that a fundamental connection exists between the low-dimensional dynamical Lie algebra, the integrability of the system, and large gradients.
Such a connection might be studied in future work.

This paper did not consider incoherent noise, which prevails in noisy quantum devices.
This type of noise is known to be another source of barren plateaus~\cite{wang2021noise}.
When the circuit is short enough, we believe that our initialization schemes can help compensate for the vanishing gradients from the incoherent noises.
Still, how the strength of the noise, the circuit depth, and the initialization schemes interplay in general is another big question that requires a separate study.
Since this is important for practical applications of variational quantum algorithms, further research on the effect of incoherent noises is necessary.

\section*{Acknowledgements}
The authors thank Modjtaba Shokrian Zini for helpful discussions and David Wierichs, Maria Schuld, and Joseph Bowles for valuable comments.
This research used resources of the National Energy Research Scientific Computing Center, a DOE Office of Science User Facility supported by the Office of Science of the U.S. Department of Energy under Contract No. DE-AC02-05CH11231 using NERSC award NERSC DDR-ERCAP0025705.
Numerical simulations were performed using \textsc{PennyLane}~\cite{bergholm2018pennylane} software package with \textsc{Lightning}~\cite{Lightning} and \textsc{Lightning-GPU}~\cite{Lightning-GPU} plugins.
The source code used for simulations is available in Ref.~\cite{park2023_github_repo}.

\medskip
\bibliographystyle{quantum}
\bibliography{references.bib}


\onecolumngrid
\appendix

\setcounter{equation}{0}%
\renewcommand{\theequation}{\thesection.\arabic{equation}}%
\setcounter{theorem}{0}%
\renewcommand{\thetheorem}{\thesection.\arabic{theorem}}%
\setcounter{corollary}{0}%
\renewcommand{\thecorollary}{\thesection.\arabic{corollary}}%

\section{Big-$O$ and related notations}\label{app:big-O}
In the main text, we have used big-$O$ and related notations.
This appendix formally defines these notations as follows:
\begin{itemize}
    \item $f(n) = \mathcal{O}(g(n))$ if there exist $n_0 \in \mathbb{N}^+$ and $c \in \mathbb{R}^+$ such that $f(n) \leq c g(n)$ for all $n \geq n_0$.
    \item $f(n) = \Omega(g(n))$ if there exist $N_0 \in \mathbb{N}^+$ and $c \in \mathbb{R}^+$ such that $f(n) \geq c g(n)$ for all $n \geq n_0$.
    \item $f(n) = \Theta(g(n))$ if $f(n) = \mathcal{O}(g(n))$ and $f(n) = \Omega(g(n))$.
\end{itemize}

\section{Long-time average of the variance of gradients in the Hamiltonian dynamics}
\label{app:long-time-avg-grad-ham-evol}

In this appendix, we prove Proposition~\ref{prop:lte_grad}, which gives the lower bound of a long-time average of the squared gradient for a Hamiltonian evolution. Direct computation of $\langle \psi_0 | i [G, O(t)]|\psi_0 \rangle^2$ gives

\begin{align}
	\langle \psi_0 | i[G, O(t)]|\psi_0  \rangle^2 &= -\Bigl[ \braket{\psi_0 | G e^{iHt} O e^{-iHt} | \psi_0} - \braket{\psi_0 |  e^{iHt} O e^{-iHt} G| \psi_0}\Bigr]^2 \\
	&= -\Bigl[ \sum_{ijk} C_i^* G_{ij} e^{i(E_j - E_k)t} O_{jk} C_k - \sum_{lmn} C_l^* e^{i(E_l - E_m)t} O_{lm} G_{mn} C_n\Bigr]^2\\
	&= -\sum_{ijki'j'k'}C_i^* G_{ij} e^{i(E_j - E_k)t} O_{jk} C_k C_{i'}^* G_{i'j'} e^{i(E_{j'} - E_{k'})t} O_{j'k'} C_{k'} \nonumber \\
	&\qquad -\sum_{lmnl'm'n'}C_l^* e^{i(E_l - E_m)t} O_{lm} G_{mn} C_n C_{l'}^* e^{i(E_{l'} - E_{m'})t} O_{l'm'} G_{m'n'} C_{n'} \nonumber \\
	&\qquad + 2 \sum_{ijklmn}C_i^* G_{ij} e^{i(E_j - E_k)t} O_{jk} C_k C_l^* e^{i(E_l - E_m)t} O_{lm} G_{mn} C_n \label{eq:time_average_gradients}
\end{align}
where $C_i = \braket{E_i|\psi_0}$, $G_{ij} = \braket{E_i | G | E_j}$, and $O_{jk} = \braket{E_j | O | E_k}$.

We assume that the Hamiltonian $H$ satisfies the non-degenerate energy-gap condition, i.e.,
\begin{align}
    E_i - E_j = E_k - E_l \text{\quad iff \quad} \begin{cases}
        &i = k \text{ and } j = l \\
        &i = j \text{ and } k = l
    \end{cases}.
\end{align}
Under this condition, averaging Eq.~\eqref{eq:time_average_gradients} over time yields
\begin{align}
	\lim_{T \rightarrow \infty} \frac{1}{T}\int_{0}^{T} dt \Bigl\{ - \langle [G, O(t)] \rangle^2 \Bigr\} &= 2 \sum_{ijkn} C_i^* G_{ij} O_{jk} |C_k|^2 O_{kj} G_{jn} C_n - \sum_{lmnn'} C_l^* O_{lm} G_{mn} C_n C_m^* O_{ml} G_{ln'} C_{n'}  \nonumber \\
	&\qquad- \sum_{ijkl} C_i^* G_{ij} O_{jk} C_k C_l^* G_{lk}O_{kj} C_j - \braket{\psi| [G, \tilde{O}] |\psi}^2 \\
	&= F_H(\psi, G, O) - \braket{\psi| [G, \tilde{O}] |\psi}^2  \geq F_H(\psi, G, O) \label{eq:equilibration_grad_ineq}
\end{align}
where $\tilde{O} = \sum_j O_{jj}\ket{E_j}\bra{E_j}$.
As $\braket{\psi | [G, \tilde{O}]|\psi}$ is purely imaginary, we obtain the last inequality.
The inequality in the main text is then obtained by changing the summation indices.

We also note that the eigenstate thermalization hypothesis~\cite{deutsch1991quantum,srednicki1994chaos,rigol2008thermalization} suggests that the last term, $\braket{\psi| [G, \tilde{O}] |\psi}^2$, is exponentially small in $N$.

\section{Generating random Hamiltonians}\label{app:gen_rand_ham}
In the main text, we numerically observed the scaling behaviors of gradient magnitudes using random Hamiltonians.
Here, we describe detailed steps to generate such random $k$-local Hamiltonians.
For a given $k$, we first create a set of terms $S = \{\sigma^{1}_{a_1}\sigma^{2}_{a_2}\cdots\sigma^{k}_{a_k}\}$ where each $\sigma^{i}_{a_i}$ is one of the the Pauli matrices at site $i$ ($\{I_i, X_i, Y_i, Z_i\}$) where $a_k \in \{0,1,2,3\}$ for all $k$.
We then remove terms duplicated under translation.
For example, as $X \otimes I \otimes I$, $I \otimes X \otimes I$, and $I \otimes I \otimes X$ generate the same terms under translation, we only keep one of them.
We then construct a random Hamiltonian $H=\sum_{s \in S} c_s \sum_{n=1}^N \mathsf{T}^n s$, where the coefficients $c_s$ are samples from the normal distribution $\mathcal{N}(0, 1)$ and $\mathsf{T}$ is the translation operator ($\mathsf{T} \sigma^{i}_a = \sigma^{i+1}_a$).
We also generate random time-reversal symmetric Hamiltonians ($H^* = H$) using the same method but removing purely imaginary operators (that contain odd numbers of Pauli-$Y$s) from $S$.

\section{Approximation of the HVA from the truncated Floquet-Magnus expansion}
\label{app:approximating_hva}

In this appendix, we prove Proposition~\ref{prop:effective_ham_fm} using the truncated Floquet-Magnus (FM) expansion.
The FM expansion~\cite{magnus1954exponential} provides a time-independent effective Hamiltonian for a unitary evolution from a time-dependent Hamiltonian.
While this expansion diverges for a general many-body Hamiltonian,
recent works~\cite{kuwahara2016floquet,abanin2017rigorous} have shown that we can still use the expansion after truncating high-order terms.

\subsection{Truncated Floquet-Magnus expansion}
Let us introduce the truncated FM expansion following the notation in Ref.~\cite{kuwahara2016floquet}.
We consider a system defined on a lattice with $N$ spins where each spin is labeled by $i=1$ to $N$.
The set of all spins is denoted by $\Lambda = \{1,\cdots,N\}$.
We consider a time-dependent Hamiltonian $H(t)$ defined for $0 \leq t \leq \tau$. 
We decompose the Hamiltonian into $H(t) = H_0 + V(t)$ where $H_0$ is the time-independent part and $V(t)$ is the remaining time-dependent part.
Both parts have at most $k$-body interactions (we do not impose geometric locality yet). We then write the Hamiltonian terms as 
\begin{align}
    H_0 = \sum_{|X| \leq k } h_X, \qquad V(t) = \sum_{|X| \leq k} v_X(t),
\end{align}
where $X$ is all possible subsets of $\Lambda$ and $|X|$ is the number of elements in the set.

We introduce a parameter $J$ that upper bounds local interaction strength and some additional parameters for the expansion:
\begin{align}
    \sum_{X: X \ni i} (\Vert h_X \Vert + \Vert v_X(t) \Vert ) \leq J \quad \forall i \in \Lambda,\qquad V_0 := \sum_{|X| \leq k} \frac{1}{\tau} \int_0^\tau \Vert v_X(t) \Vert dt, \qquad \lambda := 2k J. \label{eq:fm_expansion_params}
\end{align}

Under this setting, we are interested in the Floquet Hamiltonian $H_F$ defined as
\begin{align}
    e^{-iH_F \tau} := \mathcal{T}[e^{-i\int_0^\tau H(t) dt}] ,
\end{align}
where $\mathcal{T}[\cdot]$ is the time-ordering operator.
One can expand the Floquet Hamiltonian as $H_F = \sum_{n=0}^\infty \tau^n \Omega_n$ where the terms $\{\Omega_n\}_{n=0}^\infty$ are given by the FM expansion as follows:
\begin{align}
    \Omega_n &= \frac{1}{(n+1)^2} \sum_{\sigma \in S_{n+1}} (-1)^{n-\omega(\sigma)} \frac{\omega(\sigma)!(n-\omega(\sigma))!}{n!} \nonumber \\
    &\qquad \times \frac{1}{i^n \tau^{n+1}} \int_0^{\tau} dt_{n+1}\cdots \int_0^{t_3} dt_{2} \int_0^{t_2} dt_{1} [H(t_{\sigma(n+1))},[H(t_{\sigma(n)}),\cdots,[H(t_{\sigma(2)}),H(t_{\sigma(1)})]\cdots ]], \label{eq:fm_omega}
\end{align}
where $S_{n+1}$ is the permutation group on $n+1$ letters, $\omega(\sigma) = \sum_{i=1}^n \mathbf{1}[\sigma(i+1)-\sigma(i)]$, and $\mathbf{1}(x)$ is the Heaviside step function.
In our setup, where the Hamiltonian terms have at most $k$-body interactions, $\Omega_n$ has at most $(n+1)k$-body interactions.
We further have the following upper bound for $\Omega_n$ (Lemma~1 in Ref.~\cite{kuwahara2016floquet}):
\begin{align}
    \Vert \Omega_n \Vert \leq \frac{2V_0 \lambda^n}{(n+1)^2} n! =: \overline{\Omega}_n.
\end{align}

One can see that the convergence condition $\Vert \Omega_{n+1} \Vert \tau^{n+1} < \Vert \Omega_n \Vert \tau^n$ only holds up to $n \approx(\lambda \tau)^{-1}$.
Indeed, the FM expansion diverges for a many-body Hamiltonian unless $\tau$ also scales with $N$.
Even though this fact suggests that the FM expansion might not be useful,
it turned out that one can still use a truncated series $H_F^{(n)} := \sum_{m=0}^n \tau^m \Omega_m$ to describe long-time dynamics accurately:

\begin{theorem}[Theorem~1 in Ref.~\cite{kuwahara2016floquet}]
With our parameters in Eq.~\eqref{eq:fm_expansion_params} and additional condition $\tau \leq 1/(4\lambda)$, the time evolution under the time-dependent Hamiltonian $H(t) = H_0 + V(t)$ is close to that generated by the truncated Floquet Hamiltonian $H_F^{(n_0)} = \sum_{m=0}^{n_0} \Omega_m \tau^m$ with
\begin{align}
    n_0 := \left\lfloor \frac{1}{16 \lambda \tau}  \right\rfloor,
\end{align}
in the sense that
\begin{align}
    \Vert e^{-iH_F\tau} - e^{-i H_F^{(n_0)}\tau} \Vert \leq 6 V_0 \tau 2^{-n_0}.
\end{align}
\end{theorem}

However, as $n_0$ increases with $N$, if $\tau \sim N^{-\alpha}$ for a positive $\alpha$, the interaction range of $H_F^{(n_0)}$ increases with $N$.
As we want strict locality in our Hamiltonian (it must be $k'$-local for a constant $k'$), a bound for $H_F^{(n)}$ for a fixed $n$ should be useful, which is provided by the following Corollary.

\begin{corollary}[Corollary~1 in Ref.~\cite{kuwahara2016floquet}] \label{cor:truncated_fm}
Under the same condition, we have
\begin{align}
    \Vert e^{-iH_F \tau} - e^{-i H_F^{(n)}\tau} \Vert \leq 6V_0 \tau 2^{-n_0} + \overline{\Omega}_{n+1} \tau^{n+2},
\end{align}
where $H_F^{(n)} = \sum_{m=0}^n \Omega_m \tau^m$ for arbitrary $n \leq n_0$.
\end{corollary}

\subsection{Approximating the HVA}
We now consider the HVA given by 
\begin{align}
    U=\prod_{i=p}^1 e^{-iH^{(q)} \theta_{i,q}} \cdots e^{-iH^{(1)} \theta_{i,1}}
\end{align}
where we assume that each $H^{(j)}$ is the sum of commuting Pauli strings (products of Pauli operators) acting on at most $k$ geometrically local sites.
For example, $\sum_{i=1}^N X_i X_{i+1}$ from the HVA for the transverse-field Ising model satisfies this (both for the periodic and open boundary conditions) with $k=2$.

We now interpret the HVA as a time-dependent Hamiltonian given as
\begin{align}
    \tilde{H}(t) := 
    \begin{dcases}
    H^{(1)} &\text{for } 0 \leq t \leq \theta_{1,1} \\
    H^{(2)} &\text{for } \theta_{1,1} \leq t \leq \theta_{1,1}+\theta_{1,2} \\
    \cdots \\
    H^{(q)} &\text{for } \sum_{j=1}^{q-1}\theta_{1,j} \leq t \leq \sum_{j=1}^q \theta_{1,j} \\
    H^{(1)} &\text{for } \sum_{j=1}^q \theta_{1,j} \leq t \leq \sum_{j=1}^q \theta_{1,j}+\theta_{2,1} \\
    \cdots \\
    H^{(q)} &\text{for }  \sum_{i=1}^p\sum_{j=1}^{q-1} \theta_{i,j} \leq t \leq \sum_{i=1}^p \sum_{j=1}^q \theta_{i,j}
    \end{dcases}. \label{eq:hva_to_td_ham}
\end{align}

For convenience, we define $\Upsilon_{n,m} := \sum_{i=1}^{n-1}\sum_{j=1}^q \theta_{i,j} + \sum_{j=1}^m \theta_{n,j}$ which is the cumulative sum of $\{\theta\}$.
For any subcircuit of the HVA $U_b\cdots U_a$, where $a=(i,j)$ and $b=(i',j')$ are indices for the layers,
we consider $H(t) = H_0 + V(t)$ with $H_0 = 0$ and $V(t) = \tilde{H}(t + \Upsilon_{a-1})$ [Eq.~\eqref{eq:hva_to_td_ham}] defined for $0 \leq t \leq \Upsilon_{b} - \Upsilon_{a-1} := \tau$ (where we use $\Upsilon_{a-1}$ to denote the sum of the parameters before layer $a=(i,j)$ and $\Upsilon_b=\Upsilon_{i',j'}$ for $b=(i',j')$).

Parameters for the FM expansion can be obtained by writing $V(t)$ as
\begin{align}
    V(t) = \tilde{H}(t + \Upsilon_{a-1}) = \sum_{|X| \leq k} h_X(t).
\end{align}
Following the notation in the main text, we have 
\begin{align}
V_0 &= \frac{1}{\tau} \int_{0}^\tau \sum_{|X| \leq k} \Vert h_X(t) \Vert dt \leq \sup_{0 \leq t \leq \tau} \sum_{|X| \leq k} \Vert h_X(t) \Vert = \max_m \sum_{|X| \leq k} \Vert h_X^{(m)} \Vert = H_{\rm max}
\end{align}
for $V_0$ defined in Eq.~\eqref{eq:fm_expansion_params} and $H_{\rm max}$ defined in Eq.~\eqref{eq:h0_def}. 
Under this setup, applying Corollary~\ref{cor:truncated_fm} 
to $U_R$ and $U_L$ defined in the main text yields Proposition~\ref{prop:effective_ham_fm}.
Precisely, for a given $n \leq n_0 = \lfloor 1/(32kJ t_{R,L}) \rfloor$, there are $(n+1)k$-local Hamiltonians $H_R$ and $H_L$ such that 
\begin{align}
    &\Bigl\Vert U_R - e^{-i H_R t_R} \Bigr\Vert \leq 6 H_{\rm max}  2^{-\lfloor 1/(32kJ t_R) \rfloor} t_R + \frac{2 H_{\rm max} (2kJ)^{n+1}}{(n+2)^2}(n+1)! t_R^{n+2},\\
    &\Bigl\Vert U_L - e^{-i H_L t_L} \Bigr\Vert \leq 6 H_{\rm max}  2^{-\lfloor 1/(32kJ t_L) \rfloor} t_L + \frac{2 H_{\rm max} (2kJ)^{n+1}}{(n+2)^2}(n+1)! t_L^{n+2}
\end{align}
are satisfied (where we put $\lambda = 2kJ$ from Eq.~\eqref{eq:fm_expansion_params}).

Furthermore, $H_R$ and $H_L$ share any symmetries that $\{H^{(j)}\}$ have, which follows from the property of the commutator, i.e., $W[H_1,H_2]W^{-1} = W(H_1 H_2 - H_2 H_1) W^{-1} = (WH_1W^{-1}) (WH_2W^{-1}) - (WH_2W^{-1}) (WH_1W^{-1}) = [WH_1W^{-1},WH_2W^{-1}]$. Thus, for example, if all $\{H^{(j)}\}$ are translationally invariant, the resulting Hamiltonians $H_R$ and $H_L$ are also translationally invariant.

Obtaining the norm of each term of $H_R$ ($H_L$) is also possible.
For convenience, let $K$ be one of $H_R$ or $H_L$ defined by 
$K := H^{(n)}_F = \sum_{m=0}^n \Omega_m \tau^m$ for $\tau=t_R$ or $\tau = t_L$.
For each time $t$, let us define $j[t] \in \{1, \cdots, q\}$ to be the index such that $V(t) = H^{(j[t])}$.
Then $V(t_{\sigma(1)}) = H^{(j[t_{\sigma(1)}])} = \sum_X h^{(j[t_{\sigma(1)}])}_X$ where $h^{(j[t_{\sigma(1)}])}_X$ acts on at most $k$ sites.
Inserting this expression in Eq.~\eqref{eq:fm_omega} gives
\begin{align}
    \Omega_n &= \sum_X \frac{1}{(n+1)^2} \sum_{\sigma \in S_{n+1}} (-1)^{n-\omega(\sigma)} \frac{\omega(\sigma)!(n-\omega(\sigma))!}{n!} \nonumber \\
    &\qquad \times \frac{1}{i^n \tau^{n+1}} \int_0^{\tau} dt_{n+1}\cdots \int_0^{t_3} dt_{2} \int_0^{t_2} dt_{1} [H(t_{\sigma(n+1))},[H(t_{\sigma(n)}),\cdots,[H(t_{\sigma(2)}),h^{(j[t_{\sigma(1)}])}_X]\cdots ]].
\end{align}

So we write $K=\sum_X k_{\tilde{X}}$ with
\begin{align}
    &k_{\tilde{X}} = \sum_{m=0}^{n} \frac{\tau^m}{(m+1)^2} \sum_{\sigma \in S_{m+1}} (-1)^{m-\omega(\sigma)} \frac{\omega(\sigma)!(m-\omega(\sigma))!}{m!} \nonumber \\
    &\qquad \times \frac{1}{i^m \tau^{m+1}} \int_0^{\tau} dt_{m+1}\cdots \int_0^{t_3} dt_{2} \int_0^{t_2} dt_{1} [H(t_{\sigma(m+1))},[H(t_{\sigma(m)}),\cdots,[H(t_{\sigma(2)}),h^{(j[t_{\sigma(1)}])}_{X_i}]\cdots ]].
\end{align}

Locality of $k_{\tilde{X}}$ follows from the fact that the multicommutator $[H^{(i_n)},[H^{(i_{n-1})},\cdots,[H^{(1)},O]\cdots ]]$ acts on at most $(n+1)k$ nearby sites for any operator $O$ acting on at most $k$ local sites, and the Hamiltonians $H^{(1)},\cdots,H^{(n)}$ are $k$-local.
Precisely, each $k_{\tilde{X}}$ is supported by augmented sites $\tilde{X} = \{i \in \Lambda~|~\mathrm{dist}(i, X) \leq nk \}$ where $\mathrm{dist}(i, X) = \min_{j \in X} \mathrm{dist}(i,j)$.

Finally, we obtain a bound of the norm of $k_{\tilde{X}}$ using the inequality
\begin{align}
    &\biggl \Vert  \int_0^{\tau} dt_{m+1}\cdots \int_0^{t_3} dt_{2} \int_0^{t_2} dt_{1} [H(t_{\sigma(m+1))},[H(t_{\sigma(m)}),\cdots,[H(t_{\sigma(2)}),h^{(j[t_{\sigma(1)}])}_X]\cdots ]]\biggr\Vert \nonumber \\
    & \leq \frac{\tau^{n+1}}{(n+1)!} \max_{i_1,\cdots,i_n} \Vert [H^{(i_n)},\cdots,[H^{(i_2)},h^{(i_1)}_X]\cdots ]] \Vert,
\end{align}
and the following lemma.

\begin{lemma}[Consequence of Lemma~3 in Ref.~\cite{kuwahara2016floquet}]
Let $\{H^{(j)}\}$ be $k$-local and $\sum_{X:X \ni i} \Vert h^{(j)}_X \Vert \leq J$ for all $j$.
Then for an arbitrary operator $O$ supported on $k$ local sites, we have
\begin{align}
    \Vert [H^{(i_n)},[H^{(i_{n-1})},\cdots,[H^{(1)},O]\cdots ]] \Vert \leq (n!) (2kJ)^n \Vert O \Vert.
\end{align}
\end{lemma}

We thus have
\begin{align}
    \Vert k_{\tilde{X}} \Vert \leq \sum_{m=0}^n \frac{(2kJ)^m}{(m+1)^2}m! \tau^m, \label{eq:fm_term_bound}
\end{align}
where we use $|\sum_{\sigma \in S_{m+1}}| = (m+1)!$, $\theta(\sigma)!(m-\theta(\sigma))!/m! = {m \choose \theta(\sigma)}^{-1} \leq 1$, and $\Vert h^{(i_1)}_X \Vert = 1$ regardless of $j$ as $\{h^{(j)}_X\}$ are Pauli words.
As a consequence, we obtain
\begin{align*}
    \Vert K \Vert &\leq \sum_{X} \Vert k_{\tilde{X}} \Vert \leq H_{\rm max} \sum_{m=0}^n \frac{(2kJ)^m}{(m+1)^2}m! \tau^m, \numberthis \label{eq:upper_bound_K}
\end{align*}
where $H_{\rm max}$ is the maximum number of terms in $H^{(j)} = \sum_X h_X^{(j)}$ (defined in the main text).

\section{Proof of Theorem~\ref{thm:small_parameter_large_grad}}
\label{app:prove_small_parameter_large_grad}
We here provide a detailed proof showing that there exists $\tau_0 = \Theta(1/N)$ such that the HVA with $\sum_{i,j}\theta_{i,j} = t_R + t_L \leq \tau_0$ has large gradient components $\partial_{n,m}C$.
Here, we consider the cost function $C$ given by the expectation value of a local observable $O$ acting on at most $k_O$ sites and an initial state $\rho_0$ which gives $|\Tr\{\rho_0[H^{(m)},O]\}| = \Theta(1)$.

Our proof consists of three steps. First, we show that the error approximating the HVA to local Hamiltonian evolution from the FM expansion is $\mathcal{O}(1/N^2)$. Next, we derive all factors ($\Vert H_R \Vert$, $\Vert [H_L, O] \Vert $, etc.) in Proposition~\ref{prop:short_time_grad_var} from the FM expansion.
We then complete the proof by combining steps to show that there exists $\tau_0 = \Theta(1/N)$ such that $| \partial_{n,m}C |$ is lower bounded by a constant.

\subsection{Polynomially decaying bound of the error from the truncated FM expansion} 
\label{app:prove_small_parameter_large_grad_1}
Let us first analyze the error term in Proposition~\ref{prop:effective_ham_fm} for $t_R,t_L \leq c/N$ with $n=1$.
We note that $n=1$ requires $n_0 = \lfloor 1/(32kJt_{R,L}) \rfloor \geq 1$, which is satisfied for $N \geq N_0 := 32ckJ$.
In addition, we assume $kJ \geq 1$, which is true in our setting [see Eq.~\eqref{eq:def_J}].
Then, the error from the truncated FM expansion [the RHS of Eq.~\eqref{eq:fm_bound}] is given by
\begin{align}
    \epsilon &= \Bigl[ 6 c \times 2^{-\lfloor N/(32ckJ) \rfloor} + \frac{4 c^3 (2kJ)^2}{9}\frac{1}{N^2} \Bigr] \frac{H_{\rm max}}{N} \nonumber \\
    &\leq r \Bigl[ 6 c \times 2^{-\lfloor N/(32ckJ) \rfloor} + \frac{4 c^3 (2kJ)^2}{9}\frac{1}{N^2} \Bigr], \label{eq:fm_error}
\end{align}
where $r$ is a constant such that $ H_{\rm max} \leq r N$ (which is from $H_{\rm max} = \mathcal{O}(N)$).

We now use the following lemma to find $N_1$ such that the error is $\mathcal{O}(1/N^2)$ for $N \geq N_1$.
\begin{lemma}
For a given $\kappa_1,\kappa_2,\alpha > 0$ and 
\begin{align}
N_1:=\max \biggl\{ \frac{8}{\alpha}, -\frac{4}{\alpha} \log \Bigl[ \bigl( \frac{\alpha e}{8} \bigr)^2 \frac{\kappa_2}{2\kappa_1} \Bigr] \biggr\},
\end{align}
the inequality
\begin{align}
    \kappa_1 2^{-\lfloor \alpha N\rfloor} + \frac{\kappa_2}{N^2} \leq 2 \frac{\kappa_2}{N^2}
\end{align}
is satisfied for all $N \geq N_1$.
\end{lemma}
\begin{proof}
We first have
\begin{align}
    \kappa_1 2^{-\lfloor \alpha N\rfloor} - \frac{\kappa_2}{N^2} \leq 2 \kappa_1 e^{-\alpha N/2} - \frac{\kappa_2}{N^2} = \frac{2 \kappa_1 N^2 e^{-\alpha N/2} - \kappa_2}{N^2}.
\end{align}
Let us define $f(N):=2 \kappa_1 N^2 \exp[-\alpha N/2] = 2\kappa_1 \exp[-\alpha N/2 + 2 \log N]$.
For $N \geq 8/\alpha$, we have
\begin{align}
    \log N \leq \frac{\alpha}{8} N + \log \bigl[ \frac{8}{\alpha e} \bigr].
\end{align}
Thus,
\begin{align}
    f(N) \leq 2 \kappa_1 \bigl( \frac{8}{\alpha e}\bigr)^2 \exp \bigl[ -\frac{\alpha}{4}N \bigr]
\end{align}
for all $N \geq 8/\alpha$. One sees that the RHS is smaller than $\kappa_2$ when 
\begin{align}
N \geq -\frac{4}{\alpha} \log \Bigl[ \bigl( \frac{\alpha e}{8} \bigr)^2 \frac{\kappa_2}{2\kappa_1} \Bigr],
\end{align}
which completes the proof.
\end{proof}

We then apply this lemma to obtain the upper bound of Eq.~\eqref{eq:fm_error}.
Inserting $\alpha = 1/(32ckJ)$, $\kappa_1 = 6c$, and $\kappa_2 = 4c^3(2kJ)^2/9$ gives $N_1 = 128 \gamma ckJ$ where $\gamma = \log(4^7 \cdot 3^3 / e^2) \approx 11.00$.
As $N_1 \geq N_0$, we have $\epsilon \leq \beta(c)/N^2$ for all $N \geq \max \{N_0, N_1\} = 128 \gamma ckJ$ with $\beta(c) = 8c^3 r(2kJ)^2 /9$.
The obtained bounds tells us that the $U_L$ and $U_R$ in Proposition~\ref{prop:effective_ham_fm}
which appear in $\partial_{n,m}C$ approximate to local Hamiltonian evolution.
Precisely, for $U_R = e^{-i H^{(m-1)}\theta_{n,m-1}} \cdots e^{-i H^{(1)} \theta_{i,1}}$ and $U_L = e^{-i H^{(q)}\theta_{p,q}}\cdots e^{-i H^{(m)}\theta_{n,m}}$,
there are $2k$-local Hamiltonians $H_R$, $H_L$ such that 
$\Vert U_{R,L} - e^{-iH_{R,L} t_{R,L}} \Vert \leq \beta(c)/N^2$ if $t_R,t_L \leq c/N$ for $N \geq 128 \gamma ck J$ where $t_R = \theta_{1,1}+\cdots+\theta_{n,m-1}$ and $t_L = \theta_{n,m}+\cdots+\theta_{p,q}$.

\subsection{Condition of the constant for large gradients}
We next find an upper bound of $c$ (for $\tau_0 = c/N$) from the complete expression of time $t_c$ in Proposition~\ref{prop:short_time_grad_var}.
From Eq.~\eqref{eq:upper_bound_K} with the first order expansion ($n=1$), we have 
\begin{align}
    \Vert H_{R,L} \Vert \leq H_{\rm max} \Bigl( 1 + \frac{kJ}{2} t_{R,L} \Bigr).
\end{align}

Using this inequality, we find
\begin{align}
    \Vert H_{R} \Vert \leq H_{\rm max} \Bigl( 1 + \frac{kJ}{2}t_{R} \Bigr), \qquad \Vert [H_{L},O] \Vert \leq 2 l \Vert O \Vert \Bigl( 1 + \frac{kJ}{2}t_{L} \Bigr),
\end{align}
where we obtain the second inequality by combining Eq.~\ref{eq:local_ham_norms} and Eq.~\ref{eq:fm_term_bound}.
Here, $l=|\{X: [k_{\tilde{X}}, O] \neq 0\}| \geq 1$ is a constant for a given lattice, which follows from the fact that $k_{\tilde{X}}$ acts on at most $2k$ nearby sites and $O$ is a local operator. 

In addition, we have
\begin{align}
\Vert H^{(m)} \Vert \leq H_{\rm max}, \qquad \Vert [H^{(m)},O] \Vert \leq 2 s \Vert O \Vert
\end{align}
where $s=|\{X:[h_X, O] \neq 0\}| \geq 1$ is also a constant. We used the fact that each $h_X$ is a Pauli string to obtain the second inequality.

As $t_R,t_L \leq \tau_0 = c/N$ (from $t_R,t_L \geq 0$ and $t_R + t_L \leq \tau_0$), we obtain for $N \geq N_0 \geq 32c kJ$,
\begin{align}
    \Vert H_R \Vert \leq \frac{65}{64} H_{\rm max} := \mu H_{\rm max}, \qquad \Vert [H_{L},O] \Vert \leq 2 l \Vert O \Vert \times \frac{65}{64} := 2 \mu l \Vert O \Vert
\end{align}
where $\mu = 65/64$.

Using the fact that $\Vert H_{\rm max} \Vert \leq r N$, Proposition~\ref{prop:short_time_grad_var} yields
\begin{align}
    t_c \geq \frac{g}{8 \mu r \Vert O\Vert \max\{\mu l, s\}} \frac{1}{N},
\end{align}
where $g = |\Tr[\rho_0 [H^{(m)}, O]]|$ is the magnitude of the gradient when the circuit is trivial ($U_R=U_L=\mathbb{1}$).
Thus $t_R + t_L \leq t_c$ is satisfied for all $c$ such that
\begin{align}
    c \leq \frac{g}{8 \mu r \Vert O \Vert \max\{\mu l,s\}}.
\end{align}
We still note that the current condition implies large gradients only when $U_{R,L}$ are exact time-evolution operators. As there is an approximation error from the FM expansion, we consider this factor in the following subsection.

\subsection{Bounding gradient with the FM truncation error}
We introduce the following lemma to see how much an error from unitary approximation affects the gradients.
\begin{lemma}
For a density matrix $\rho \geq 0$ and $\Tr[\rho]=1$, Hermitian operators $A$, $\tilde{A}$, and unitary operators $U$, and $\tilde{U}$,
\begin{align*}
    \bigl| \Tr[U \rho U^\dagger A] - \Tr [\tilde{U} \rho \tilde{U}^\dagger \tilde{A}] \bigr| \leq \Vert A \Vert \Vert U - \tilde{U} \Vert + \Vert A - \tilde{A} \Vert + \Vert \tilde{A} \Vert \Vert U^\dagger - \tilde{U}^\dagger \Vert. \numberthis
\end{align*}
\begin{proof}
\begin{align*}
    \bigl| \Tr[U \rho U^\dagger A] - \Tr [\tilde{U} \rho \tilde{U}^\dagger \tilde{A}] \bigr| &= \bigl| \Tr[\rho (U^\dagger A U - \tilde{U}^\dagger \tilde{A}\tilde{U})] \bigr|\\
    &\leq \Vert U^\dagger A U - \tilde{U}^\dagger \tilde{A}\tilde{U} \Vert = \Vert U^\dagger A U - U^\dagger A \tilde{U} + U^\dagger A \tilde{U} - \tilde{U}^\dagger \tilde{A}\tilde{U} \Vert \\
    &\leq \Vert A \Vert \Vert U - \tilde{U} \Vert + \Vert U^\dagger A - \tilde{U}^\dagger \tilde{A} \Vert \\
    &\leq \Vert A \Vert \Vert U - \tilde{U} \Vert + \Vert A - \tilde{A} \Vert + \Vert \tilde{A} \Vert \Vert U^\dagger - \tilde{U}^\dagger \Vert.
\end{align*}
\end{proof}
\end{lemma}

Let us now apply this lemma to $\partial_{n,m} C = \Tr\{U_R \rho_0 U_R^\dagger [H^{(m)}, U_L^\dagger O U_L]\}$ with $U = U_R$, $\tilde{U}=e^{-iH_R t_R}$, $A = [H^{(m)}, U_L^\dagger O U_L]$, and $\tilde{A}=[H^{(m)}, e^{i H_L t_L} O e^{-iH_L t_L} ]$.
Denoting $\epsilon$ by the error from the FM expansion, i.e.,
$\Vert e^{-i H_R t_R} - U_R \Vert \leq \epsilon$ and $\Vert e^{-i H_L t_L} - U_L \Vert \leq \epsilon$, we obtain
\begin{align*}
    &\bigl| \Tr\{U_R \rho_0 U_R^\dagger [H^{(m)}, U_L^\dagger O U_L]\} - \Tr\{e^{-i H_R t_R} \rho_0 e^{i H_R t_R} [H^{(m)}, e^{i H_L t_L} O e^{-i H_L t_L}]\} \bigr| \\
    &\leq \epsilon  \Vert [H^{(m)}, U_L^\dagger O U_L] \Vert + \Vert [H^{(m)}, U_L^\dagger O U_L] - [H^{(m)}, e^{i H_L t_L} O e^{-i H_L t_L}]\Vert + \epsilon \Vert [H^{(m)}, e^{i H_L t_L} O e^{-i H_L t_L}] \Vert \\
    &\leq 4 \epsilon \Vert H^{(m)} \Vert \Vert O \Vert + 2 \Vert H^{(m)} \Vert \Vert U_L^\dagger O U_L - e^{i H_L t_L} O e^{-i H_L t_L}\Vert \\
    &\leq 8\epsilon \Vert H^{(m)} \Vert \Vert O \Vert , \numberthis
\end{align*}
where we used $\Vert [A,B] \Vert \leq 2 \Vert A \Vert \Vert B \Vert$ to obtain the third line and $\Vert UOU^\dagger - \tilde{U}O \tilde{U}^\dagger \Vert \leq 2 \Vert O \Vert \Vert U - \tilde{U}\Vert$ for the last inequality.
As we obtained a bound $\epsilon \leq \beta(c)/N^2$ for $N \geq N_1(c)$ (final result in Sec.~\ref{app:prove_small_parameter_large_grad_1}) and $\Vert H^{(m)} \Vert \leq H_{\rm max} \leq rN$, the error is upper bounded by $8 r \beta(c)\Vert O \Vert/N$ for a sufficiently large $N$.
Thus, for $N \geq 32 r \beta(c) \Vert O \Vert /g$, we can bound the error to be less than $g/4$.
As Proposition~\ref{prop:short_time_grad_var} implies $|\Tr\{e^{-i H_R t_R} \rho_0 e^{i H_R t_R} [H^{(m)}, e^{i H_L t_L} O e^{-i H_L t_L}]\}| \geq g/2$, we have $|\Tr\{U_R \rho_0 U_R^\dagger [H^{(m)}, U_L^\dagger O U_L]\}| \geq g/4$ under this condition.

We summarize the overall result as follows.
For the HVA for $N$ qubits, a local operator $O$ and an initial state $\rho_0$ are given.
We assume there is a constant $g > 0$ and $m$ such that $|\Tr\{\rho_0 [H^{(m)}, O]\}| \geq g$ regardless of $N$.
Then we fix
\begin{align}
    c = \frac{g}{8 \mu r \Vert O \Vert \max\{\mu l,s\}},
\end{align}
where $r,l$, $s$ are constants obtained from the properties of $\{H^{(m)}\}$, and $\mu=65/64$.
Then for $N_{\rm min} = \max\{128\gamma ckJ, 32 r \beta(c) \Vert O\Vert/g \}$,
\begin{align}
    \Bigl| \frac{\partial C}{\partial \theta_{n,m}} \Bigr| \geq \frac{1}{4}g
\end{align}
is satisfied for all $N \geq N_{\rm min}$ if $t_{R} + t_{L} \leq c/N$.

\section{Vanishing gradient after a finite time evolution}
\label{app:finite_time_evolution_zero_grad}
In the main text and previous Appendix, we argued that there exists $\tau_0 = \Theta(1/N)$ such that the HVA with constraints $\theta_{i,j} \geq 0$ and $\sum_{i,j} \theta_{i,j} \leq \tau_0$ does not have vanishing gradients if $|\Tr[\rho_0 [H^{(m)}, O]]| \neq 0$ for some $m$.
In this subsection, we provide an example whose gradient component vanishes where the sum of parameters is a constant.
This implies that if there is $\tilde{\tau}_0$ such that the gradient is bounded by a constant when $\sum_{ij} \theta_{i,h} \leq \tilde{\tau}_0$, $\tilde{\tau}_0$ must be smaller than this constant.

We consider the Ising model with transverse and longitudinal fields whose Hamiltonian is given by $\mathcal{H} = - \sum_i Z_i Z_{i+1} - h \sum_i X_i - g\sum_i Z_i$.
The HVA for this model can be written as
\begin{align}
    \ket{\psi(\{\theta_{i,j}\})} = \prod_{i=p}^1 e^{-i \theta_{i,3}\sum_i Z_i} e^{-i \theta_{i,2}\sum_i X_i} e^{-i \theta_{i,1}\sum_i Z_i Z_{i+1}} \ket{+}^{N}
\end{align}
where the initial state $\rho_0 = \ket{\psi_0}\bra{\psi_0}$ with $\ket{\psi_0} = \ket{+}^{\otimes N}$.
Consider a local observable $O=Y_1$ and gradient for $\theta_{p,3}$ which is given by
\begin{align}
    \partial_{p,3}C = i\Tr\bigl\{ U_R \rho_0 U_R^\dagger [\sum_i Z_i, Y_1] \bigr\} = -2 \Tr \bigl\{ U_R \rho_0 U_R^\dagger X_1 \bigr\}.
\end{align}
As $|\Tr[\rho_0[\sum_i Z_i, Y_1]]| = 2$, Theorem~\ref{thm:small_parameter_large_grad} implies that there is $\tau_0 = \Theta(1/N)$ such that any parameters satisfying $\theta_{i,j} \geq 0$ and $\sum_{i,j} \theta_{i,j} \leq \tau_0$ give an $\Theta(1)$ gradient.
We now consider a parameter set with $\theta_{i,2}=\theta_{i,1}=0$ for all $i$. This gives $\ket{\psi(\{\theta_{i,j}\})} = U_R \ket{\psi_0} = (\cos \Upsilon \ket{+} - i \sin \Upsilon \ket{-})^{\otimes N}$ where $\Upsilon := \sum_{i,j} \theta_{i,j} = \sum_{i}\theta_{i,3}$.
Thus for $\Upsilon = \pi/4$, we obtain $\partial_{p,3}C = -2 \bra{y;+}^{\otimes N} X_1 \ket{y;+}^{\otimes N} = 0$.
This implies that we do not expect that the condition of the theorem is relaxed to $\tau_0 \geq \pi/4 = \Theta(1)$.

\end{document}